%% file: secrecy-outage-probability-bounds.tex
\documentclass[10pt,english,journal,twocolumn]{IEEEtran}
\usepackage{blindtext}

\usepackage{babel}
\usepackage[babel]{microtype}
\usepackage[T1]{fontenc}
\usepackage[keeplastbox]{flushend}

\usepackage{refcount}

\usepackage{graphicx}
\usepackage{xcolor}
\usepackage{tikz}
\usepackage{pgfplots}

\usepackage{tabularx}
\usepackage{booktabs}
\usepackage{stfloats}

\usepackage{amsmath}
\usepackage{amsfonts}
\usepackage{amssymb}
\usepackage{amsthm}
\usepackage{bm}
\usepackage{siunitx}

\usepackage{csquotes}
\usepackage[backend=biber,url=false,style=ieee,isbn=false,doi=false]{biblatex}%
\bibliography{literature.bib}
\renewcommand*{\bibfont}{\footnotesize}

\usepackage{hyperref}
\hypersetup{
	pdfauthor={Karl-Ludwig Besser},
	pdftitle={Bounds on the Secrecy Outage Probability for Dependent Fading Channels},
	colorlinks=true,
	allcolors=.,
	urlcolor=blue,
}

\usepackage[acronyms,nomain,xindy]{glossaries}
\makeglossaries
\glsenableentrycount
\loadglsentries{acronyms.tex}
\setacronymstyle{long-short}

\input{setups.tex}

\title{Bounds on the Secrecy Outage Probability for Dependent Fading Channels}
\author{Karl-Ludwig Besser, \IEEEmembership{Student Member, IEEE} and Eduard A. Jorswieck, \IEEEmembership{Fellow, IEEE}
\thanks{The authors are with the Institute of Communications Technology, Technische Universit\"at Braunschweig, 38106 Braunschweig, Germany (email: \{{k.besser}, {e.jorswieck}\}@tu-bs.de).}
\thanks{This work is supported in part by the German Research Foundation (DFG) under grant JO\,801/23-1.}}

\begin{document}
\maketitle

\begin{abstract}
	The amount of sensitive data, which is transmitted wirelessly will increase with future technologies. This raises many questions about secure data transmission. Besides cryptography, information-theoretic security gained increasing attention over the recent years. Among others, it deals with the problem of secure data transmission on the physical layer to a legitimate receiver (Bob) in the presence of an eavesdropper (Eve).
	In this work, we investigate upper and lower bounds on the secrecy outage probability for slowly-fading wiretap channels with an arbitrary dependency structure between the fading channels to Bob and Eve. Both cases of absence of {channel-state information at the transmitter (CSI-T)} and availability of CSI-T of only the main channel to the legitimate receiver are considered. Furthermore, we derive explicit expressions for the upper and lower bounds for Rayleigh fading and compare them to the case of independent channels. The joint distribution of the legitimate and eavesdropper channels has a tremendous impact on the achievable secrecy outage probability.
	The bounds enable developing guaranteed secrecy schemes by only measuring the marginal channel distributions.
\end{abstract}


\input{introduction.tex}
\input{preliminaries.tex}
\input{outage-bounds-main-csit.tex}
\input{outage-bounds-no-csit.tex}
\input{outage-full-outage.tex}
\input{rayleigh-fading.tex}
\input{conclusion.tex}

\appendices
\input{app-proofs.tex}

\printbibliography
\end{document}

%% file: acronyms.tex
\newacronym{5g}{5G}{Fifth-Generation Wireless Systems}
\newacronym{ask}{ASK}{amplitude-shift keying}
\newacronym[plural={AWGN}, \glsshortpluralkey={AWGN}]{awgn}{AWGN}{additive white Gaussian noise}

\newacronym{ber}{BER}{bit error rate}
\newacronym{bler}{BLER}{block error rate}
\newacronym{bpsk}{BPSK}{binary phase-shift keying}

\newacronym{cdf}{CDF}{cumulative distribution function}
\newacronym{cnn}{CNN}{convolutional neural network}
\newacronym{csi}{CSI}{channel-state information}
\newacronym{csir}{CSI-R}{channel-state information at the receiver}
\newacronym{csit}{CSI-T}{channel-state information at the transmitter}

\newacronym{ecc}{ECC}{error-correcting code}

\newacronym{ff}{FF}{feed-forward}
\newacronym{ffnn}{FF-NN}{feed-forward neural network}
\newacronym{fsk}{FSK}{frequency-shift keying}

\newacronym{gm}{GM}{Gaussian mixture}

\newacronym{iid}{i.\,i.\,d.}{independent and identically distributed}
\newacronym{il}{IL}{information leakage}
\newacronym{iot}{IoT}{internet-of-things}

\newacronym{lb}{LB}{lower bound}

\newacronym{map}{MAP}{maximum a posteriori}
\newacronym{ml}{ML}{machine learning}
\newacronym{mop}{MOP}{multi-objective programming problem}
\newacronym{mse}{MSE}{mean squared error}

\newacronym{nn}{NN}{neural network}
\newacronym{nve}{NVE}{normalized validation error}

\newacronym{pdf}{PDF}{probability density function}
\newacronym{pwc}{PWC}{polar wiretap code}

\newacronym{qam}{QAM}{quadrature amplitude modulation}

\newacronym{rbf}{RBF}{radial basis function}
\newacronym{relu}{ReLU}{rectified linear unit}
\newacronym{rnn}{RNN}{recurrent neural network}

\newacronym{sgd}{SGD}{stochastic gradient descent}
\newacronym{sgp}{SGP}{secure goodput}
\newacronym{snr}{SNR}{signal-to-noise ratio}
\newacronym{svm}{SVM}{support vector machine}

\newacronym{tvd}{TVD}{total variation distance}

\newacronym{ub}{UB}{upper bound}
\newacronym{urllc}{URLLC}{ultra-reliable low latency communication}

%% file: setups.tex
\pgfplotsset{compat=newest}
\pgfplotsset{plot coordinates/math parser=false}

\usetikzlibrary{patterns, arrows, circuits, trees, positioning, fit, external}
\usepgfplotslibrary{patchplots}
\tikzset{%
	block/.style = {draw, thick, rectangle, minimum height=2em, minimum width=3em},
	sum/.style = {draw, thick, circle, inner sep=.1em, minimum height=2em}, %
}

\newcommand{\diff}{\ensuremath{\,\mathrm{d}}}
\newcommand{\expect}[2][]{\ensuremath{\mathbb{E}_{#1}\left[#2\right]}}

\newcommand{\abs}[1]{\ensuremath{\left|#1\right|}}

\newcommand{\positive}[1]{\ensuremath{\left[#1\right]^{+}}}
\newcommand{\leqone}[1]{\ensuremath{\left[#1\right]^{\leq 1}}}

\newcommand{\ratesecret}{\ensuremath{R_{S}}}
\newcommand{\rateconfusion}{\ensuremath{R_{d}}}
\newcommand{\snrbob}{\ensuremath{\rho_{x}}}
\newcommand{\snreve}{\ensuremath{\rho_{y}}}
\newcommand{\X}{\ensuremath{\bm{X}}}
\newcommand{\Xt}{\ensuremath{\bm{\tilde{X}}}}
\newcommand{\Y}{\ensuremath{\bm{Y}}}
\newcommand{\Yt}{\ensuremath{\bm{\tilde{Y}}}}
\newcommand{\xt}{\ensuremath{\tilde{x}}}
\newcommand{\yt}{\ensuremath{\tilde{y}}}
\newcommand{\Z}{\ensuremath{\bm{Z}}}
\newcommand{\lx}{\ensuremath{{\lambda_{x}}}}
\newcommand{\ly}{\ensuremath{{\lambda_{y}}}}
\newcommand{\lxt}{\ensuremath{\tilde{\lambda}_{x}}}
\newcommand{\lyt}{\ensuremath{\tilde{\lambda}_{y}}}
\newcommand{\eventsum}{\ensuremath{E_{1}}}
\newcommand{\eventbob}{\ensuremath{E_{2}}}
\newcommand{\eventeve}{\ensuremath{E_{3}}}

\theoremstyle{plain}%
\newtheorem{thm}{Theorem}%

\newtheorem{cor}{Corollary}

\theoremstyle{remark}
\newtheorem{rem}{Remark}

\definecolor{plot0}{HTML}{CF171A}
\definecolor{plot2}{HTML}{9D4390}
\definecolor{plot3}{HTML}{EB811B}
\definecolor{plot4}{HTML}{14B03D}
\definecolor{plot1}{HTML}{2db7d2}

\definecolor{coolblack}{rgb}{0.0, 0.18, 0.39}

\definecolor{change}{HTML}{000000}%

%% file: introduction.tex
\section{Introduction}\label{sec:introduction}

Wireless data transmission plays already today an important role in personal communications. With new emerging concepts and technologies like \gls{iot} and 6G~\cite{Saad2019}, the amount of wirelessly transmitted data will further increase. This also includes private and sensitive data {which} should therefore be transmitted securely.

Cryptography is one option to achieve this goal. However, it requires a shared-key between the different parties which does not scale well with the massive increase of devices.
Another approach is physical layer security~\cite{Bloch2011} where the physical properties of the wireless channel are exploited to enable a secure transmission. One fundamental result shows that certain channel codes exist for the standard \gls{awgn} channel, which allow zero information leakage to a passive eavesdropper~\cite{Wyner1975}.
In a setting where the channels to the legitimate receiver and the eavesdropper also experience fading, outages can occur due to the random nature of the fading~\cite{Gungor2013}. 

A designer of a communication system in such a scenario might face the following task: \emph{My application can tolerate a maximum secrecy outage probability $\varepsilon$. What is the maximum transmission rate such that the secrecy outage probability of my system is below $\varepsilon$, even in the worst-case?}
Especially for critical tasks, it is important to guarantee {a maximum outage probability} even in the worst-case. The uncertainty about the system can stem from many different factors. The one we consider in this work is the statistical dependence of the communication channels. The marginal distributions of the fading coefficients for different positions (or users) can easily be measured. However, the dependency between them is usually unclear because the attacker might choose position and orientation according to the legitimate transmitter and receiver.

On the other hand, relaying technologies as well as new emerging technologies like reconfigurable intelligent surfaces~\cite{DiRenzo2019} might enable a pro-active design of the radio environment and therefore a control of the dependency structure. 
Another way of enhancing the secrecy outage probability is the use of cooperative jamming~\cite{Hu2019}.
A wireless engineer might therefore ask himself in the future: \emph{How close is the operating point in terms of secrecy outage rate of my system to the best-possible and how can I achieve the upper bound?}
In this case, the lower bound on the secrecy outage probability over all joint distributions can be used as a benchmark.

In the literature, usually independent legitimate and eavesdropper channels are considered~\cite{Bloch2008b,Gopala2008,Bhargav2016,Wang2018,Zhao2019}. However, real measurements demonstrate that the commonly used assumption of independent channels is not always justified~\cite{Lee1973,Sang-BinRhee1974}. In particular, if the eavesdropper chooses its position and orientation carefully. 
Therefore, some work has already studied the influence of correlation between channels on the secrecy performance. In \cite{HyoungsukJeon2011}, bounds on the secrecy capacity for correlated fading channels are provided. However, the results hold only for the high \gls{snr} regime and only for the specific case of Rayleigh fading. Additionally, the not well justified assumption of full \gls{csit} about both channels is assumed. The results of \cite{HyoungsukJeon2011} are extended in \cite{Sun2012} for general \gls{snr}. A different fading distribution, namely log-normal fading, is considered in \cite{Liu2013}. Additional differences are that \cite{Liu2013} focuses on the outage probability instead of the secrecy capacity and no \gls{csit} about the eavesdropper's channel is assumed.
The secrecy outage probability of correlated channels with small-scale fading and shadowing has been investigated in~\cite{Alexandropoulos2018}.

In all of this previous work only (positive) linear correlation is considered. In contrast, the upper and lower bounds that we derive in this work are based on copulas~\cite{Nelsen2006} and hold for any dependency between the channels.
{For peaceful systems, this has already been considered in \cite{Besser2020}.} In the area of physical layer security, a similar approach has been used to derive general bounds on the ergodic secret-key capacity~\cite{Besser2020wsa}.

In this work, we answer the aforementioned question about the secrecy outage probability of the worst-case\footnote{\label{footnote:best-worst-case}{The worst- and best-case refer to the upper and lower bound on the secrecy outage probability, respectively, over all joint distributions with given marginals.}} with respect to the joint distribution of the fading coefficients of the channels to the legitimate receiver and a passive eavesdropper. In addition, we will derive the outage probability for the best-case\footnotemark[\getrefnumber{footnote:best-worst-case}]. All possible joint distributions cause an outage probability somewhere between the two bounds. The bounds are achievable, i.e., there exist joint distributions which indeed achieve them with equality. 

We always consider only statistical \gls{csit} of the eavesdropper channel. However, we consider the two cases of perfect \gls{csit} about the channel to {the} legitimate receiver and only statistical \gls{csit}.
The results are explicitly evaluated for the case of Rayleigh fading and compared to the case of independent channels.
The main contributions are summarized in the following.
\begin{itemize}
	\item We derive tight upper and lower bounds on the secrecy outage probability for dependent fading channels when the transmitter Alice has perfect \gls{csit} about the channel to the legitimate receiver Bob.
	\item We derive tight upper and lower bounds on the secrecy outage probability for dependent fading channels when the transmitter Alice has only statistical \gls{csit}.
	\item We give a general sufficient condition on the distributions of the fading coefficients for which the secrecy outage probability is only determined by the quality of Bob's channel.
	\item We evaluate all bounds for the special scenario of Rayleigh fading and compare them to the case of independent channels. All numerical evaluations and plots are made publicly available in an interactive notebook at \cite{BesserGitlab}.
\end{itemize}

The rest of the paper is organized as follows. In Section~\ref{sec:prelimiaries}, the system model, preliminaries on physical layer security, and the required mathematical background are introduced.
The bounds on the secrecy outage probability for the cases of perfect \gls{csit} and only statistical \gls{csit} about the main channel are derived in Sections~\ref{sec:outage-csit-main} and \ref{sec:outage-no-csit}, respectively.
In Section~\ref{sec:full-outage}, we consider an alternative definition for the secrecy outage probability and derive the bounds for this notion.
All of these results are then evaluated explicitly and compared to the case of independent channels for Rayleigh fading in Section~\ref{sec:rayleigh}. Finally, Section~\ref{sec:conclusion} concludes the paper.

\textit{Notation:}
Throughout this work, we use following notation. Random variables are denoted in capital boldface letters, e.g., $\X$, and their realizations in small letters, e.g., $x$. We will use $F$ and $f$ for a probability distribution and its density, respectively. The expectation is denoted by $\mathbb{E}$ and the probability of an event by $\Pr$. It is assumed that all considered distributions are continuous. The dual of a copula $C$ is written as $\bar{C}$.
As a shorthand, we use $\positive{x}=\max\left[x, 0\right]$ and similarly $\leqone{x}=\min\left[x, 1\right]$.
The derivative of a univariate function $g$ is written as $g^{\prime}$.
The real numbers are denoted by $\mathbb{R}$. Logarithms, if not stated otherwise, are assumed to be with respect to the natural base.

%% file: preliminaries.tex
\section{System Model, Problem Statement and Preliminaries}\label{sec:prelimiaries}
In this section, we will first state the system model and some important definitions and facts from the area of physical layer security~\cite{Bloch2011}. Next, we state the problem formulation. Afterwards, we will introduce needed mathematical background from copula theory~\cite{Nelsen2006}.

\subsection{Fading Wiretap Channel}
Throughout this work, we consider the discrete-time Gaussian wiretap channel with quasi-static block flat-fading as channel model~\cite{Bloch2008b}. %
The transmitter Alice wants to transmit messages $\bm{M}$ securely to the legitimate receiver Bob {by encoding them into codewords $\bm{A}$}. The communication channel between them suffers from fading $\bm{H}_{x}$ and \gls{awgn} $\bm{N}_{x}$ with average noise power $N_{B}$. We will refer to this channel as \enquote{main channel}.
The received signal at Bob at time $i$ is given by
\begin{equation*}
\bm{B}(i) = \bm{H}_{x}(i)\bm{A}(i) + \bm{N}_{x}(i)\,.
\end{equation*}

The communication between Alice and Bob is eavesdropped by Eve. The channel between this passive eavesdropper and Alice suffers from fading $\bm{H}_{y}$ and \gls{awgn} $\bm{N}_{y}$ with average noise power $N_{E}$.
The received signal at Eve at time $i$ is given by
\begin{equation*}
\bm{E}(i) = \bm{H}_{y}(i)\bm{A}(i) + \bm{N}_{y}(i)\,.
\end{equation*}

The transmission at Alice is subject to an average power constraint
\begin{equation*}
\frac{1}{N}\sum_{i=1}^{N}\expect{\abs{\bm{A}(i)}^2}\leq P\,.
\end{equation*}
The receiver \glspl{snr} are given as $\snrbob=P/N_{B}$ and $\snreve=P/N_{E}$ for Bob and Eve, respectively.

Since we consider quasi-static block flat-fading, the fading coefficients are constant for the transmission of one codeword, i.e., $\bm{H}_{x}(i) = \bm{H}_{x}$ and $\bm{H}_{y}(i) = \bm{H}_{y}$. We will therefore drop the time index $i$ in the following.

The goal of Alice is to transmit the message in such a way that Bob can decode it reliably while Eve cannot infer any information about it.
One way to achieve this is to use a wiretap code with binning structure~\cite{Bloch2011}. The idea is the following:
{Alice wants to encode secret messages of length $n\ratesecret$ into codewords of length $n$. Instead of a one-to-one mapping, each of the $2^{n\ratesecret}$ messages can be mapped to one of $2^{n\rateconfusion}$ possible codewords, which are referred to as a bin.} A stochastic encoder randomly selects one of the possible codewords from the bin corresponding to the message which is transmitted. {Therefore, the total number of codewords is $2^{n(\ratesecret+\rateconfusion)}$, which can be seen as transmitting a message of length $n(\ratesecret+\rateconfusion)$ which is constructed from a secret part of length $n\ratesecret$ and a random part of length $n\rateconfusion$.} The rates $\ratesecret$ and $\rateconfusion$ of the secret and dummy messages, respectively, are selected in such a way that Bob is able to decode them reliably. If Eve's channel is worse than Bob's, she is only able to decode a bin but not the actual secure message.
However, in the case of fading channels, the qualities of the channels to Bob and Eve vary due to the random nature of the fading coefficients.
The instantaneous secrecy capacity $C_{S}$ is then a random variable and it is given by~\cite[Lem.~1]{Bloch2008b}
\begin{equation}
C_S = \positive{\log_2\left(1 + \snrbob\X\right) - \log_2\left(1 + \snreve\Y\right)}\,,
\end{equation}
with the shorthands $\X=\abs{\bm{H}_x}^2$ and $\Y=\abs{\bm{H}_y}^2$.
If this secrecy capacity is less than the secrecy rate $\ratesecret$, which is used for the transmission, a secrecy outage occurs. Another event which causes an outage is the case that Bob cannot decode the message reliably~\cite[Rem.~5.7]{Bloch2011}. This happens in the event that his channel capacity is less than the total transmission rate $\rateconfusion+\ratesecret$. If Alice has perfect \gls{csi} about the main channel, {she can use rate adaption~\cite[Sec.~5.2.3]{Bloch2011} to avoid the event that the transmission rate falls below the capacity of the main channel.}
We will consider the cases of statistical \gls{csit} to Eve and either perfect \gls{csit} or statistical \gls{csit} to Bob in Sections~\ref{sec:outage-csit-main} and \ref{sec:outage-no-csit}, respectively.

\subsection{Problem Formulation}
With the above considerations, we can give the exact problem formulation for the rest of the work, as follows.

We consider the previously described wiretap channel with quasi-static block flat-fading. The predefined secure communication rate is $\ratesecret$. There is only statistical \gls{csit} of the channel to Eve. If the transmitter has perfect \gls{csi} about the main channel, the secrecy outage probability is defined as~\cite[Def.~5.1]{Bloch2011}
\begin{equation}\label{eq:def-secrecy-outage-main-csit}
\varepsilon_{\text{CSIT}} = \Pr\nolimits_{\X, \Y}\left(C_S(\X, \Y) < \ratesecret\right)\,.
\end{equation}
If no \gls{csit} is available, the outages due to decoding errors at Bob need to be included. With the rate $\rateconfusion$ of dummy messages, i.e., a total transmission rate $\rateconfusion+\ratesecret$, {the definition of the secrecy outage probability} can be modified to~\cite[Rem.~5.7]{Bloch2011}
\begin{equation}
\varepsilon_{\text{no}} = \Pr\nolimits_{\X, \Y}\left(C_S(\X, \Y) < \ratesecret \vee C_m(\X)<\rateconfusion+\ratesecret\right),
\end{equation}
where $C_m$ is the capacity of the main channel.

{It can be seen that the outage probabilities depend on the joint distribution of $\X$ and $\Y$. It is possible that the channels are not independent and usually, we only have information about the marginal distributions of the channel gains. Our considered problem statement is therefore as follows:}
Find the best- and worst-case outage probabilities for $\varepsilon_{\text{CSIT}}$ and $\varepsilon_{\text{no}}$ over all possible joint distributions $F_{\X,\Y}$ given the marginal distributions $F_{\X}$, $F_{\Y}$, i.e., 
\begin{equation*}
\inf_{F_{\X,\Y}: \genfrac{}{}{0pt}{}{F_{\X}(x) = F_{\X,\Y}(x,\infty)}{F_{\Y}(y) = F_{\X,\Y}(\infty,y)}} \varepsilon \leq \varepsilon \leq \sup_{F_{\X,\Y}: \genfrac{}{}{0pt}{}{F_{\X}(x) = F_{\X,\Y}(x,\infty)}{F_{\Y}(y) = F_{\X,\Y}(\infty,y)}} \varepsilon\,.
\end{equation*}
In Sections~\ref{sec:outage-csit-main} and \ref{sec:outage-no-csit}, we will derive upper and lower bounds on these outage probabilities over all possible joint distributions $F_{\X,\Y}$.

\subsection{Mathematical Background}
The bounds on the outage probability are derived from copula theory~\cite{Nelsen2006}. One major advantage is that this covers all possible dependency structures between the channels and not only linear correlation.
Our results, which we will present in Sections~\ref{sec:outage-csit-main} and \ref{sec:outage-no-csit}, are based on \cite[Thm.~1]{Williamson1990} which in turn immediately follows from \cite[Thm.~3.1]{Frank1987}. We will restate \cite[Thm.~1]{Williamson1990} in the following as Theorem~\ref{thm:bounds-sum-frank}. Since we will use the idea of its proof for our results, we also restate the proof.
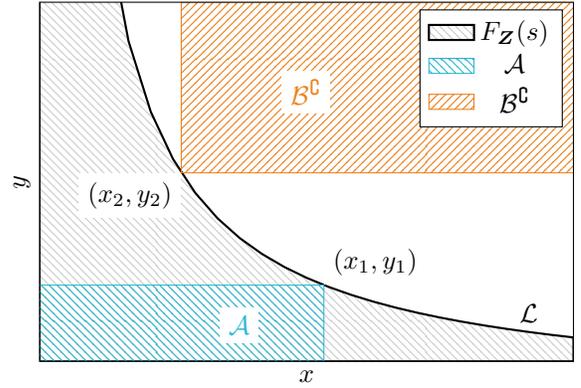
\begin{figure}
		\centering
		\input{img/prelim-proof-bounds-sum.tex}
		
		\vspace*{-0.5em}
		\caption{Visualization of the different regions used for the proof of Theorem~\ref{thm:bounds-sum-frank} (adapted from \cite[Fig.~1]{Williamson1990}). The line $\mathcal{L}$ is given by $L(x, y)=s$.}
		\label{fig:proof-copula-bounds-l}
\end{figure}
\begin{thm}[{\cite[Thm.~1]{Williamson1990}}]\label{thm:bounds-sum-frank}
	Let $\X$ and $\Y$ be random variables over the non-negative real numbers with \gls{cdf} $F_{\X}$ and $F_{\Y}$, respectively. Let $L$ be a binary operation that is non-decreasing in each place and continuous. The \gls{cdf} of the random variable $\Z=L(\X, \Y)$ is bounded by
	\begin{equation}
	\tau_{W}(F_{\X}, F_{\Y}) \leq F_{\Z} \leq \phi_{W}(F_{\X}, F_{\Y})\,,
	\end{equation}
	with
	\begin{align}
	\tau_{C}(F_{\X}, F_{\Y})(s) &= \sup_{L(x, y)=s} C(F_{\X}(x), F_{\Y}(y))\\
	\phi_{C}(F_{\X}, F_{\Y})(s) &= \inf_{L(x, y)=s} \bar{C}(F_{\X}(x), F_{\Y}(y))
	\end{align}
	for a copula $C$ and its dual $\bar{C}(a, b)=a+b-C(a, b)$.
\end{thm}
\begin{proof}
	For the proof of the theorem, we first need the well-known Fr\'{e}chet-Hoeffding bound~\cite{Nelsen2006}
	\begin{equation*}
	W(a, b) = \max\left[a+b-1, 0\right] \leq C(a, b)\,.
	\end{equation*}
	Next, we define the line $\mathcal{L}=\left\{(x, y)\;|\;L(x, y)=s\right\}$, which is exemplary shown in Fig.~\ref{fig:proof-copula-bounds-l}.
	With reference to Fig.~\ref{fig:proof-copula-bounds-l}, we can observe the following for any pairs $(x_1, y_1)$ and $(x_2, y_2)$ on line $\mathcal{L}$,
	\begin{align}
	W(F_{\X}(x_1), F_{\Y}(y_1)) &\leq C(F_{\X}(x_1), F_{\Y}(y_1))\\
	&= \iint_{\mathcal{A}}\diff{C(F_{\X}(x), F_{\Y}(y))}\\
	&\leq F_{\Z}(s)\\
	&\leq \iint_{\mathcal{B}}\diff{C(F_{\X}(x), F_{\Y}(y))}\label{eq:proof-less-b}\\
	\begin{split}&=F_{\X}(x_2) + F_{\Y}(y_2) \\&\qquad- C(F_{\X}(x_2), F_{\Y}(y_2))\end{split}\label{eq:proof-expression-b}\\
	&= \bar{C}(F_{\X}(x_2), F_{\Y}(y_2))\\
	&\leq \bar{W}(F_{\X}(x_2), F_{\Y}(y_2))\,.
	\end{align}
	
	The first inequality is the Fr\'{e}chet-Hoeffding lower bound. The second line is the definition of the joint \gls{cdf} of $\X$ and $\Y$. The third line follows from the fact, that the probability of $\Z<s$ is the gray shaded area below $\mathcal{L}$. It can be seen that this is smaller than the area $\mathcal{A}$. On the other hand, the gray area is a subset of $\mathcal{B}$, which gives \eqref{eq:proof-less-b}. The integral in \eqref{eq:proof-less-b} can be expressed as the sum of the individual probabilities, which is given as \eqref{eq:proof-expression-b}. This is defined to be the dual of copula $C$. The last line again follows immediately from the Fr\'{e}chet-Hoeffding lower bound.
\end{proof}

%% file: img/prelim-proof-bounds-sum.tex
\begin{tikzpicture}
\begin{axis}[
	width=.98\linewidth,
	height=.26\textheight,
	xlabel={$x$},
	ylabel={$y$},
	domain=0.01:6,
	xmin=0.01,
	xmax=6,
	ymin=0.1,
	ymax=1.1,
	axis on top,
	xtick=\empty,
	ytick=\empty,
]
\addplot[black,thick,pattern=north west lines, pattern color=gray!40, area legend] {1/x}\closedcycle;
\addlegendentry{$F_{\Z}(s)$};

\addplot[plot1,pattern=north west lines, pattern color=plot1, domain=0:3.2, area legend] {1/3.2}\closedcycle;
\addlegendentry{$\mathcal{A}$};
\node[anchor=center, text=plot1, fill=white] at (axis cs: 2.25, 0.2) {$\mathcal{A}$};

\draw[plot3,pattern=north east lines, pattern color=plot3] (1.6,0.625) rectangle (6,1.1);
\addlegendimage{area legend,pattern=north east lines,pattern color=plot3,draw=plot3}
\addlegendentry{$\mathcal{B}^{\complement}$};
\node[anchor=center, text=plot3, fill=white] at (axis cs: 3, 0.85) {$\mathcal{B}^{\complement}$};

\node[anchor=south west, text=black, fill=white] at (axis cs: 3.2, 0.3125) {$(x_1, y_1)$};
\node[anchor=north east, text=black, fill=white] at (axis cs: 1.6, 0.625) {$(x_2, y_2)$};
\node[anchor=south, text=black] at (axis cs: 5.5, 0.1818) {$\mathcal{L}$};

\end{axis}
\end{tikzpicture}

%% file: outage-bounds-main-csit.tex
\section{Bounds on the Secrecy Outage Probability with Perfect CSI-T}\label{sec:outage-csit-main}
At first, we assume that Alice has perfect \gls{csi} about the main channel to the legitimate receiver Bob.
In this case, a secrecy outage happens, if the instantaneous secrecy capacity is less than the secrecy rate $\ratesecret$ used in the transmission. This is represented by the following event
\begin{equation*}
	\eventsum:\quad \log_2(1+\snrbob\X)-\log_2(1+\snreve\Y) < \ratesecret
\end{equation*}
with random fading channel gains $\X=\abs{\bm{H}_x}^2$ and $\Y=\abs{\bm{H}_y}^2$.
An equivalent formulation, which we will use in the following, is
\begin{equation}\label{eq:def-event-sum}
\eventsum:\quad \Xt+\Yt < 2^{\ratesecret} - 1
\end{equation}
with the random variables $\Xt=\snrbob\X$ and $\Yt=-2^{\ratesecret}\snreve\Y$.
Therefore, the secrecy outage probability $\varepsilon$ is given as~\cite{Bloch2011}
\begin{equation}\label{eq:def-outage}
	\varepsilon_{\text{CSIT}} = \Pr(\eventsum) = \Pr\left(\Xt+\Yt < s\right)\,,
\end{equation}
where we introduce the shorthand $s=2^{\ratesecret}-1$.

Given the joint \gls{cdf} $F_{\Xt, \Yt}$ of $\Xt$ and $\Yt$, this probability is given as
\begin{equation}
\varepsilon = \int_{\xt+\yt<s}\diff{F_{\Xt, \Yt}(\xt, \yt)}\,.
\end{equation}

With the well-known Fr\'{e}chet-Hoeffding bounds, we can bound this probability as shown in \cite[Thm.~3.1]{Frank1987} or in Theorem \ref{thm:bounds-sum-frank} for the special case of $L(x,y)=x+y$ as
\begin{align}
\underline{\varepsilon} = \inf_{F_{\Xt, \Yt}} \varepsilon &= \sup_{\xt+\yt=s} \positive{F_{\Xt}(\xt) + F_{\Yt}(\yt) - 1}\label{eq:lower-outage-sum}\\
\overline{\varepsilon} = \sup_{F_{\Xt, \Yt}} \varepsilon &= \inf_{\xt+\yt=s} \leqone{F_{\Xt}(\xt)+F_{\Yt}(\yt)} \label{eq:upper-outage-sum}\,.
\end{align}

In the following, we will take a closer look at solutions and conditions for these bounds.

\subsection{Lower Bound}
We will state the result for the lower bound first and give the derivation in the following.
\begin{thm}[Lower Bound on the Secrecy Outage Probability with Main \gls{csit}]\label{thm:lower-bound-general-csit}
	Let $\X$ and $\Y$ be random variables over the non-negative real numbers representing the squared magnitude of the channel gains to Bob and Eve, respectively. The transmitter has perfect \gls{csi} only about the main channel to Bob and no \gls{csi} about the channel to Eve.
	The secrecy outage probability is then lower bounded by
	\begin{equation}\label{eq:lower-bound-general-csit}
	\underline{\varepsilon_{\text{CSIT}}} = \begin{cases}
	F_{\Xt}(s) & \text{\upshape if } f_{\Yt}^{\prime}(\yt^\star) + f_{\Xt}^{\prime}(s-\yt^\star) \geq 0\\
	\max_{\yt^\star \in \mathcal{Y}} g\left(\yt^{\star}\right) & \text{\upshape if } f_{\Yt}^{\prime}(\yt^\star) + f_{\Xt}^{\prime}(s-\yt^\star) < 0
	\end{cases},
	\end{equation}
	with
	\begin{equation}\label{eq:def-g-thm}
	g(\yt) = F_{\Xt}(s-\yt) + F_{\Yt}(\yt) - 1\,,
	\end{equation}
	and where the maximum is over all $\yt^\star\leq0$. The set of feasible $\yt^\star$ is the following
	\begin{equation}
		\mathcal{Y} = \left\{\yt^\star \;\middle|\; \yt^\star < 0 \wedge f_{\Yt}(\yt^{\star}) = f_{\Xt}(s-\yt^{\star})\right\} \cup \left\{0\right\}\,.
		\label{eq:cons}
	\end{equation}
	The used shorthands are $\Xt=\snrbob\X$, $\Yt=-2^{\ratesecret}\snreve\Y$, and $s=2^{\ratesecret}-1$.
\end{thm}
\begin{proof}
	{The proof can be found in Appendix~\ref{app:proof-thm-lower-bound-general-csit}.}
\end{proof}

Note that the solution $\underline{\varepsilon}=F_{\Xt}(s)$ is only depending on the quality of the channel to Bob. Therefore, from an operational point of view, this is the best-possible lower bound and similar to the case when there is no eavesdropper. In order to have this solution, it is sufficient that $g(\yt)\leq F_{\Xt}(s)$ for all $\yt<0$. From this, the following corollary follows immediately.
\begin{cor}\label{cor:lower-bound-no-eve}
	Let the random variables $\Xt$ and $\Yt$ be as previously defined. The transmitter has perfect \gls{csi} only about the main channel and only statistical \gls{csi} about the eavesdropper channel.
	In this case, the secrecy outage probability is given as $\underline{\varepsilon} = F_{\Xt}(s)$ if
	\begin{equation}\label{eq:condition-best-lower-csit}
	\Pr\left(s<\Xt<s-\yt\right) \leq \Pr\left(\Yt\geq\yt\right)
	\end{equation}
	holds for all $\yt\leq0$.
\end{cor}
In Section~\ref{sec:rayleigh}, we will see that the phenomenon discussed in Corollary~\ref{cor:lower-bound-no-eve} can occur, e.g., when both the main channel and the eavesdropper channel observe Rayleigh fading.

\subsection{Upper Bound}\label{sub:upper-main-csit}
\begin{thm}[Upper Bound on the Secrecy Outage Probability with Main \gls{csit}]\label{thm:upper-bound-general-csit}
	Let $\X$ and $\Y$ be random variables over the non-negative real numbers representing the squared magnitude of the channel gains to Bob and Eve, respectively. The transmitter has perfect \gls{csi} about the main channel to Bob.
	The secrecy outage probability is then upper bounded by
	\begin{equation}\label{eq:upper-bound-general-csit}
	\overline{\varepsilon_{\text{CSIT}}} = \begin{cases}
	\min_{\yt^\star \in \mathcal{Z}} h(\yt^\star) & \text{\upshape if } f_{\Yt}^{\prime}(\yt^\star) + f_{\Xt}^{\prime}(s-\yt^\star) \geq 0\\
	1 & \text{\upshape if } f_{\Yt}^{\prime}(\yt^\star) + f_{\Xt}^{\prime}(s-\yt^\star) < 0
	\end{cases},
	\end{equation}
	with
	\begin{equation}
	h(\yt) = F_{\Xt}(s-\yt) + F_{\Yt}(\yt)\,
	\end{equation}
	and	where the minimum is over all $\yt^\star$ which are from the following feasible set
	\begin{equation}
	 \mathcal{Z} = \left\{\yt^\star \;\middle|\; \yt^\star < 0 \wedge f_{\Yt}(\yt^{\star}) = f_{\Xt}(s-\yt^{\star})\right\} \cup \{-\infty, 0\}\,.
	\end{equation}
\end{thm}
\begin{proof}
	{The proof can be found in Appendix~\ref{app:proof-thm-upper-bound-general-csit}.}
\end{proof}
\subsection{Independent Channels}
For comparison, we also derive the outage probability in the case of independent $\X$ and $\Y$.
In this case, the joint \gls{pdf} is the product of the marginal ones, i.e., $F_{\Xt, \Yt}=F_{\Xt}F_{\Yt}$. Therefore, the outage probability can be calculated as
\begin{align}
\varepsilon_{\text{ind, CSIT}} &= \int\limits_{-\infty}^{0} \int\limits_{0}^{s-\yt} f_{\Xt}(\xt) f_{\Yt}(\yt)\diff{\xt}\diff{\yt}\\
&=\int\limits_{-\infty}^{0} f_{\Yt}(\yt) F_{\Xt}(s-\yt) \diff{\yt}\label{eq:outage-independent-main-csit}\,.
\end{align}

{This general expression will be evaluated explicitly for the case of Rayleigh fading in \eqref{eq:outage-rayleigh-independent-main-csit} in Section~\ref{sec:rayleigh}.}

%% file: outage-bounds-no-csit.tex
\section{Statistical Channel State Information at the Transmitter}\label{sec:outage-no-csit}
After we derived upper and lower bounds on the secrecy outage probability for the scenario where perfect \gls{csit} about the main channel is available, we now drop this assumption in this section. 
In the following, we only consider perfect \gls{csir} and only statistical \gls{csit} of legitimate and eavesdropper channels.

In this case, a secrecy outage event not only occurs when the secrecy rate $\ratesecret$ is too high, i.e., event $\eventsum$ from \eqref{eq:def-event-sum}, but also when Bob is not able to decode~\cite[Rem.~5.7]{Bloch2011}. This event $\eventbob$ is now also possible since Alice does not have perfect \gls{csit} and therefore cannot apply power or rate adaption.
The total transmission rate which Bob is supposed to support is $R=\rateconfusion+\ratesecret$, where $\rateconfusion$ is the rate of the dummy messages which are used to confuse the eavesdropper. Therefore, the event $\eventbob$ is given as
\begin{equation}\label{eq:def-event-bob}
	\eventbob:\quad \log_2\left(1+\snrbob\X\right)< \rateconfusion+\ratesecret \;\Leftrightarrow\; \Xt < 2^{R}-1\,,
\end{equation}
where we again use the shorthand $\Xt=\snrbob\X$.

The overall secrecy outage probability in the case of only statistical \gls{csit} is then given by
\begin{equation}
\varepsilon = \Pr(\eventsum \cup \eventbob) = \int\limits_{\mathcal{S}_1\cup\mathcal{S}_2}\!\diff{F_{\Xt, \Yt}(\xt, \yt)}\,.
\end{equation}
The probability is equal to the integral of the joint distribution over the area corresponding to the event $\eventsum\cup\eventbob$. The areas $\mathcal{S}_{1}$ and $\mathcal{S}_{2}$ corresponding to the events $\eventsum$ and $\eventbob$, respectively, are shown in Fig.~\ref{fig:areas-no-csit}. The line $\mathcal{L}=\left\{(\xt, \yt)\,\middle|\,\xt=\max[2^{R}-1, 2^{\ratesecret}-1-\yt]\right\}$ denotes the border of the integration area.
We can define the boundary by $L(\xt, \yt)=\xt-\max\left[2^{\rateconfusion+\ratesecret}-1, 2^{\ratesecret}-1-\yt\right]$ which is non-decreasing in each place. We can therefore apply Theorem~\ref{thm:bounds-sum-frank}. This can also be easily seen from Fig.~\ref{fig:areas-no-csit}.
\begin{figure}
	\centering
	\input{img/areas-no-csit.tex}
	
	\vspace*{-0.5em}
	\caption{Visualization of the areas of the outage event with only statistical \gls{csit}. The areas $\mathcal{S}_{1}$ and $\mathcal{S}_2$ correspond to the single events $\eventsum$ and $\eventbob$, respectively. The line $\mathcal{L}$ denotes the border of the area of the event $\eventsum\cup\eventbob$.}
	\label{fig:areas-no-csit}
\end{figure}
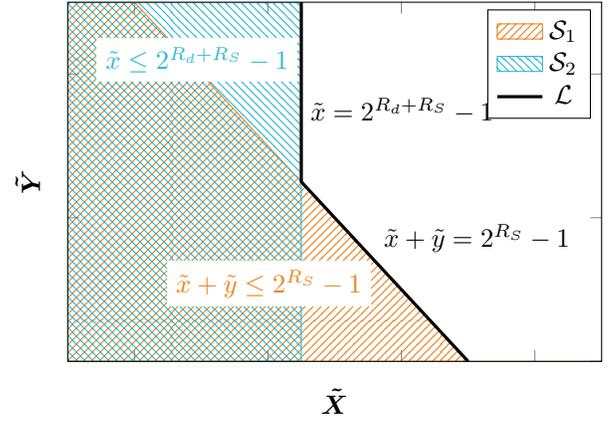

With Theorem~\ref{thm:bounds-sum-frank}, we get the following optimization problems that we need to solve to get the bounds on the secrecy outage probability for the scenario that only statistical \gls{csit} is available
\begin{align}
\underline{\varepsilon_{\text{no}}} = \inf_{F_{\Xt, \Yt}} \varepsilon &= \sup_{\mathcal{L}} \positive{F_{\Xt}(\xt)+F_{\Yt}(\yt)-1}\label{eq:lower-bound-general-no-csit}\\
\overline{\varepsilon_{\text{no}}} = \sup_{F_{\Xt, \Yt}} \varepsilon &= \inf_{\mathcal{L}} \leqone{F_{\Xt}(\xt)+F_{\Yt}(\yt)}\label{eq:upper-bound-general-no-csit}\,.
\end{align}
In the following, we will take a closer look at these upper and lower bounds and compare them to the scenario from Section~\ref{sec:outage-csit-main} where perfect \gls{csit} was available.

\subsection{Lower Bound}\label{sub:lower-no-csit}
\begin{thm}[Lower Bound on the Secrecy Outage Probability without Perfect \gls{csit}]\label{thm:lower-bound-general-no-csit}
	Let $\X$ and $\Y$ be random variables over the non-negative real numbers representing the squared magnitude of the channel gains to Bob and Eve, respectively. The transmitter has only statistical \gls{csi}.
	The secrecy outage probability is then lower bounded by
	\begin{equation}\label{eq:lower-bound-no-csit}
		\underline{\varepsilon_{\text{no}}} = \max\left[F_{\Xt}(t), \max_{\yt^\star \in \mathcal{A}} g_1(\yt^\star)\right]
	\end{equation}
	with
	\begin{equation}
		g_1(\yt) = F_{\Xt}(s-\yt) + F_{\Yt}(\yt) - 1\,,
	\end{equation}
	and	where the maximum is over all $\yt^\star$ from the set
	\begin{equation}
	\mathcal{A} = \left\{\yt^\star \;\middle|\; \yt^\star<s-t \wedge f_{\Yt}(\yt^{\star}) = f_{\Xt}(s-\yt^{\star})\right\}\,,
	\end{equation}
	and the shorthands are $s=2^{\ratesecret}-1$ and $t=2^{\rateconfusion+\ratesecret}-1$.
\end{thm}
\begin{proof}
	{The proof can be found in Appendix~\ref{app:proof-thm-lower-bound-general-no-csit}.}
\end{proof}

Analogue to Corollary~\ref{cor:lower-bound-no-eve}, we can state the following sufficient condition for the lower bound without perfect \gls{csit}.
\begin{cor}\label{cor:lower-bound-no-eve-no-csit}
	Let the random variables $\Xt$ and $\Yt$ be as previously defined. The transmitter has only statistical \gls{csi}.
	In this case, the secrecy outage probability is given as $\underline{\varepsilon_{\text{no}}} = F_{\Xt}(t)$ if
	\begin{equation}
	F_{\Xt}(s-\yt) - F_{\Xt}(t) \leq 1 - F_{\Yt}(\yt)
	\end{equation}
	holds for $\yt<s-t$.
\end{cor}

\subsection{Upper Bound}
For the upper bound on $\varepsilon$ we can combine the techniques from Sections~\ref{sub:upper-main-csit} and \ref{sub:lower-no-csit}.

\begin{thm}[Upper Bound on the Secrecy Outage Probability without Perfect \gls{csit}]\label{thm:upper-bound-general-no-csit}
	Let $\X$ and $\Y$ be random variables over the non-negative real numbers representing the squared magnitude of the channel gains to Bob and Eve, respectively. The transmitter has only statistical \gls{csi}.
	The secrecy outage probability is then upper bounded by
	\begin{equation}\label{eq:upper-bound-no-csit}
	\overline{\varepsilon_{\text{no}}} = \min\left[F_{\Xt}\left(t\right) + F_{\Yt}\left(s-t\right), \min_{\yt^\star \in \mathcal{A}} h_1(\yt^\star), 1\right]\,.
	\end{equation}
	with
	\begin{equation}
	h_1(\yt) = F_{\Xt}(s-\yt) + F_{\Yt}(\yt)\,,
	\end{equation}
	and	where the minimum is over all $\yt^\star$ from the set
	\begin{equation}\label{eq:condition-yt-star-no-csit}
	\mathcal{A} = \left\{\yt^\star \;\middle|\;\yt^\star<s-t \wedge f_{\Yt}(\yt^{\star}) = f_{\Xt}(s-\yt^{\star})\right\}\,.
	\end{equation}
\end{thm}
\begin{proof}
	{The proof can be found in Appendix~\ref{app:proof-thm-upper-bound-general-no-csit}.}
\end{proof}

\begin{rem}\label{rem:upper-bound-no-csit-single-ytstar}
	Note that $h_1(s-t)=F_{\Xt}(t)+F_{\Yt}(s-t)$ holds. If there exists at most one $\yt^\star$ in the feasible set \eqref{eq:condition-yt-star-no-csit}, we get a simplified version of \eqref{eq:upper-bound-no-csit} %
	\begin{equation*}
	\overline{\varepsilon_\text{no}} = \leqone{h_1\left(\min\left[\yt^\star, s-t\right]\right)}\,.
	\end{equation*}
\end{rem}

\begin{rem}
	The upper bound $\overline{\varepsilon}$ corresponds to the worst-case scenario. With reference to the introduction, this is an important design guideline for a system designer. Especially in critical applications, the system should be designed in such a way that the security can be guaranteed even in the worst-case.
\end{rem}

\subsection{Independent Channels}
In general, the outage probability can be calculated as the integral of the joint distribution $F_{\Xt, \Yt}$ over the region $\mathcal{S}_1\cup\mathcal{S}_2$ from Fig.~\ref{fig:areas-no-csit}. In the case of independent channels, the joint distribution is given as the product of the marginals. Thus, the outage probability can be calculated as \eqref{eq:outage-independent-no-csit} {at the bottom of the page}.
\begin{figure*}[b]
	\noindent\rule{\textwidth}{.5pt}
	\begin{equation}
	\varepsilon_{\text{ind, no}} = \int_{0}^{2^{\rateconfusion+\ratesecret}-1} f_{\Xt}(\xt)\diff{\xt} + \int\limits_{-\infty}^{2^{\ratesecret}-2^{\rateconfusion+\ratesecret}}\int\limits_{2^{\rateconfusion+\ratesecret}-1}^{2^{\ratesecret}-1-\yt} f_{\Xt}(\xt) f_{\Yt}(\yt) \diff{\xt} \diff{\yt}\label{eq:outage-independent-no-csit}\,.
	\end{equation}
\end{figure*}

{The solution for the specific case of Rayleigh fading, will be derived in Section~\ref{sub:rayleigh-no-csit}.}

%% file: img/areas-no-csit.tex
\begin{tikzpicture}
\begin{axis}[
	width=.98\linewidth,
	height=.26\textheight,
	xmin=-5,
	xmax=3,
	ymin=0,
	ymax=5,
	axis on top,
	xlabel={$\Xt$},
	ylabel={$\Yt$},
	xticklabels={,,},
	yticklabels={,,},
]
	\addplot[plot3,pattern=north east lines, pattern color=plot3, area legend] {-x+1}\closedcycle;
	\addlegendentry{$\mathcal{S}_{1}$};
	
	\addplot[plot1,pattern=north west lines, pattern color=plot1, domain=-5:-1.5, area legend] {5}\closedcycle;
	\addlegendentry{$\mathcal{S}_{2}$};

	\node[anchor=east, text=plot1, fill=white] at (axis cs: -1.55, 4.2) {$\xt\leq2^{\rateconfusion+\ratesecret}-1$};
	\node[anchor=north east, text=plot3, fill=white] at (axis cs: -.45, 1.4) {$\xt+\yt\leq2^{\ratesecret}-1$};
	
	\addplot[very thick, black] table {
		-1.5	5
		-1.5	2.5
		1	0
	};
	\addlegendentry{$\mathcal{L}$};
	\node[anchor=south west, text=black] at (axis cs: -.4, 1.4) {$\xt+\yt = 2^{\ratesecret}-1$};
	\node[anchor=south west, text=black] at (axis cs: -1.5, 3.2) {$\xt = 2^{\rateconfusion+\ratesecret}-1$};

\end{axis}
\end{tikzpicture}

%% file: outage-full-outage.tex
\section{Alternative Secrecy Outage Definition}\label{sec:full-outage}
The definition of a secrecy outage based on events $\eventsum$ and $\eventbob$ from \eqref{eq:def-event-sum} and \eqref{eq:def-event-bob}, respectively, does not take events into account when Eve is able to decode parts of the secure messages.
If we also take this into account, we get an additional outage event $\eventeve$, which is given as~\cite{Zhou2011}
\begin{equation*}
\eventeve:\quad \log_2\left(1+\snreve\Y\right) > \rateconfusion\,,
\end{equation*}
or equivalently
\begin{equation}\label{eq:def-event-eve}
\eventeve:\quad -2^{\ratesecret}\snreve\Y = \Yt < 2^{\ratesecret} - 2^{\rateconfusion+\ratesecret} = s-t\,,
\end{equation}
where we use the same definitions for $\Yt$, $s$, and $t$ as in the previous sections.

\subsection{Perfect CSIT about the Main Channel}
When perfect \gls{csit} about the main channel is available, we extend the definition of the outage event from \eqref{eq:def-event-sum} by \eqref{eq:def-event-eve} to get the alternative outage event
\begin{equation}\label{eq:def-event-full-outage-main-csit}
E_{\text{alt,CSIT}} = \eventsum \cup \eventeve\,.
\end{equation}
Note that this corresponds to
\begin{equation*}
\xt+\yt \leq s \vee \yt \leq s-t\,,
\end{equation*}
which is similar to the previously considered outage event in Section~\ref{sec:outage-no-csit}. The results from Section~\ref{sec:outage-no-csit} can be adapted by exchanging $\xt$ and $\yt$ and adapting the boundaries.
\begin{thm}\label{thm:bounds-general-full-outage-main-csit}
	Let $\X$ and $\Y$ be random variables over the non-negative real numbers representing the squared magnitude of the channel gains to Bob and Eve, respectively. The transmitter has perfect \gls{csi} about the main channel to Bob.
	The secrecy outage probability defined according to \eqref{eq:def-event-full-outage-main-csit} is then lower bounded by
	\begin{equation}%
	\underline{\varepsilon_{\text{alt,CSIT}}} = \max\left[F_{\Xt}(s), F_{\Yt}(s-t), \max_{\xt^\star \in \mathcal{B}} g_1(\xt^\star)\right]
	\end{equation}
	with
	\begin{equation}
	g_1(\xt) = F_{\Xt}(\xt) + F_{\Yt}(s-\xt) - 1\,,
	\end{equation}
	and	where the maximum is over all $\xt^\star$ from the set
	\begin{equation}
	\mathcal{B} = \left\{\xt^\star \;\middle|\; s\leq\xt^\star<t \wedge f_{\Xt}(\xt^{\star}) = f_{\Yt}(s-\xt^{\star})\right\}\,,
	\end{equation}
	and the shorthands are $s=2^{\ratesecret}-1$ and $t=2^{\rateconfusion+\ratesecret}-1$.
	The upper bound is given as
	\begin{equation}%
	\overline{\varepsilon_{\text{alt,CSIT}}} = \min\left[F_{\Xt}\left(t\right) + F_{\Yt}\left(s-t\right), \min_{\xt^\star \in \mathcal{B}} h_1(\xt^\star), 1\right]\,.
	\end{equation}
	with
	\begin{equation}
	h_1(\xt) = F_{\Xt}(\xt) + F_{\Yt}(s-\xt)\,.
	\end{equation}
\end{thm}
\begin{proof}
	From \eqref{eq:def-event-full-outage-main-csit} and with reference to Fig.~\ref{fig:areas-no-csit}, it is easy to see that the theorem follows from Theorems~\ref{thm:lower-bound-general-no-csit} and \ref{thm:upper-bound-general-no-csit} by simply exchanging $\xt$ and $\yt$ and adjusting the boundaries according to the events $\eventsum$ and $\eventeve$.
\end{proof}

\subsection{Statistical CSIT about the Main Channel}
When taking $\eventeve$ into account, we get the following definition for the secrecy outage event
\begin{equation}\label{eq:def-event-full-outage-no-csit}
E_{\text{alt,no}} = \eventsum \cup \eventbob \cup \eventeve = \eventbob \cup \eventeve\,,
\end{equation}
when only statistical \gls{csit} is available.
The probability of this event can be bounded as follows.
\begin{thm}\label{thm:bounds-general-full-outage-no-csit}
	Let $\X$ and $\Y$ be random variables over the non-negative real numbers representing the squared magnitude of the channel gains to Bob and Eve, respectively. The transmitter has only statistical \gls{csi}.
	The secrecy outage probability defined according to \eqref{eq:def-event-full-outage-no-csit} is then bounded by
	\begin{align}
		\overline{\varepsilon_{\text{alt,no}}} &= \bar{W}(F_{\Xt}(t), F_{\Yt}(s-t))\\
		\underline{\varepsilon_{\text{alt,no}}} &= \bar{M}(F_{\Xt}(t), F_{\Yt}(s-t))\,,
	\end{align}
	where $\bar{C}(a, b) = a + b - C(a, b)$ is the dual of the copula $C$.
\end{thm}
\begin{proof}
	{The proof can be found in Appendix~\ref{app:proof-thm-bounds-general-full-outage-no-csit}.}
\end{proof}

%% file: rayleigh-fading.tex
\section{Bounds for Rayleigh Fading}\label{sec:rayleigh}
We now consider the example of Rayleigh fading, i.e., $\X\sim\exp(\lx)$ and $\Y\sim\exp(\ly)$ where $\lx$ and $\ly$ are the inverse means of the channel gains, i.e., $\expect{\X} = 1/\lx$ and $\expect{\Y} = 1/\ly$. 
The auxiliary variables are then distributed according to $\Xt\sim\exp(\lx/\snrbob)$ and $-\Yt\sim\exp(\ly/(2^{\ratesecret}\snreve))$. Their \glspl{cdf} are $F_{\Xt}(\xt) = 1-\exp(-\xt\lxt)$ and $F_{\Yt}(\yt) = \exp(\yt\lyt)$, with the (inverse) scale parameters $\lxt = \lx/\snrbob$ and $\lyt = \ly/(2^{\ratesecret}\snreve)$.
Note that these distributions have monotone densities and have only at most one stationary point $\yt^\star$.

\subsection{Perfect Channel State Information at the Transmitter}
We start with the assumption from Section~\ref{sec:outage-csit-main}, that Alice has perfect \gls{csi} about the channel to Bob.

\subsubsection{Lower Bound}\label{sub:lower-rayleigh-main-csit}
To determine the stationary point $\yt^\star$, we need the derivatives of the function $g$ from \eqref{eq:def-g-thm}
\begin{align}
g(\yt) &= \exp(\lyt\yt) - \exp(\lxt(\yt-s))\label{eq:g-rayleigh}\\
g'(\yt) &= \lyt \exp(\lyt \yt) - \lxt \exp(\lxt(\yt-s))\\
g''(\yt) &= \lyt^2 \exp(\lyt \yt) - \lxt^2 \exp(\lxt(\yt-s))\,.
\end{align}
Now, we solve \eqref{eq:condition-yt-star} to obtain
\begin{equation}\label{eq:yt-star-rayleigh}
\yt^\star = \frac{\lxt s + \log\frac{\lyt}{\lxt}}{\lxt-\lyt}\,.
\end{equation}
\begin{figure}%
	\centering
	\input{img/lower-bound-monotonic.tex}
	
	\vspace*{-0.5em}
	\caption{Illustration of the function $g$ for the lower bound on the secrecy outage probability in the case of Rayleigh fading with different combinations of $\lxt$ and $\lyt$. A value $g(\yt^\star)>F_{\Xt}(s)$ is achieved at $\yt^{\star}$, if $g$ has a maximum and $\yt^\star<0$. Otherwise, the maximum is at $\yt=0$.}
	\label{fig:lower-bound-monotonic}
\end{figure}
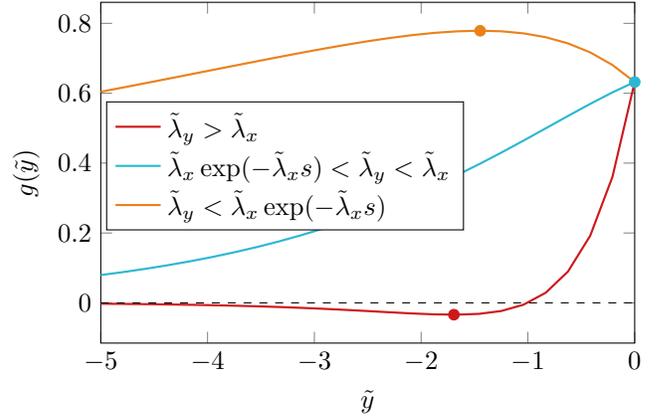
By evaluating $g''(\yt^\star)$, we get that $g$ has a maximum, if $\lxt>\lyt$. In addition, $\yt^\star<0$ needs to hold in order to have a lower bound different from $F_{\Xt}(s)$. The different possibilities are illustrated in Fig.~\ref{fig:lower-bound-monotonic}. The case that $g$ has a maximum at $\yt^\star<0$, if $\lyt<\lxt\exp(-\lxt s)$.
Combining these results yields the lower bound on the secrecy outage probability for Rayleigh fading as
\begin{equation}\label{eq:lower-rayleigh-csit-main}
\underline{\varepsilon_{\text{CSIT}}} = \begin{cases}
g\left(\frac{\lxt s + \log\frac{\lyt}{\lxt}}{\lxt-\lyt}\right) & \text{for } \lyt < \lxt\exp(-\lxt s)\\
1 - \exp(-\lxt s) & \text{else}
\end{cases}.
\end{equation}

The optimal joint distribution which achieves this lower bound is given in form of a copula in \cite[Thm.~3.2]{Frank1987}. A numerically determined joint \gls{pdf} $f_{\X, \Y}$ corresponding to this, is presented in Fig.~\ref{fig:joint-pdf-rayleigh-lower-main-csit}. The parameters are set to $\snrbob=\snreve=\SI{0}{\decibel}$, $\lx=\ly=1$, and $\ratesecret=1$. This plot nicely shows the following intuition behind the optimal coupling of the fading gains. First note that outages $\eventsum$ occur either if the main channel is bad, i.e., $\X$ is small, or if the quality of the eavesdropper channel is very good, i.e., $\Y$ is high. Therefore, the optimal joint distribution is designed in such a way that both occur simultaneously. In other words, if the main channel is bad, it is no problem if the eavesdropper has a very good channel, since an outage is likely to occur anyway. However, if the gain $\X$ to Bob is high, the gain $\Y$ should be low to avoid outages. This is exactly what can be seen in Fig.~\ref{fig:joint-pdf-rayleigh-lower-main-csit}. An interactive version of this plot can also be found at \cite{BesserGitlab} where an interested reader can change the parameters.
\begin{figure}
	\centering
	\input{img/joint-pdf-rayleigh-lower-main-csit.tex}
	
	\vspace*{-0.5em}
	\caption{Joint \gls{pdf} $f_{\X, \Y}$ that achieves the lower bound from \eqref{eq:lower-rayleigh-csit-main} for Rayleigh fading with perfect \gls{csit} about the main channel. The parameters are $\snrbob=\snreve=\SI{0}{\decibel}$, $\lx=\ly=1$, and $\ratesecret=1$.}
	\label{fig:joint-pdf-rayleigh-lower-main-csit}
\end{figure}

\begin{rem}
	As mentioned in Corollary~\ref{cor:lower-bound-no-eve}, from a system designer's point of view, it might be more interesting to have a condition on the \gls{snr} of the eavesdropper's channel for which the best lower bound can be achieved. Using the fact that $\lyt=\ly/(2^{\ratesecret}\snreve)$, we get
	\begin{equation}\label{eq:lower-rayleigh-condition-snr-eve-main-csit}
	\snreve < \frac{\ly}{\lx} \frac{\snrbob}{2^{\ratesecret}} \exp\left(\frac{\lx}{\snrbob}\left(2^{\ratesecret}-1\right)\right)\,.
	\end{equation}
	If Eve's \gls{snr} $\snreve$ is less than this value, the lower bound reduces to $F_{\Xt}(s)$, which is the lowest possible. This means that there exists a dependency structure for which the outage probability is the same as if there was no eavesdropper.
	{With the assumption $\lx=\ly=1$, \eqref{eq:lower-rayleigh-condition-snr-eve-main-csit} can be approximated as $\snreve<\snrbob$ for small rates $\ratesecret$. For arbitrary $\ratesecret$, we can find the following approximation of Eve's \gls{snr} in \si{\decibel}-scale in the high \gls{snr} regime, 
	\begin{equation*}
	\snreve^{[\si{\decibel}]} < \snrbob^{[\si{\decibel}]} - 3\ratesecret\,,
	\end{equation*}
	where we use the additional approximation $10\log_{10}(2)\approx 3$. We can therefore state that the maximum eavesdropper \gls{snr} $\snreve$, for which the secrecy outage probability can be independent of the eavesdropper, is approximately growing linearly with the \gls{snr} of the main channel $\snrbob$ in \si{\decibel}-scale. Interactive plots illustrating these approximations can be found online at~\cite{BesserGitlab}.}
	
	Unfortunately, {\eqref{eq:lower-rayleigh-condition-snr-eve-main-csit}} cannot be solved for $\ratesecret$ in a closed-form expression. However, we can use the fact that $\exp(\lxt s)\geq1$ and thus $\underline{\varepsilon}=F_{\Xt}(s)$ holds for $\lyt>\lxt$. This gives the following sufficient condition on the secrecy rate $\ratesecret$
	\begin{equation}\label{eq:cond-rs-best-lower-bound}
	\ratesecret<\log_2\left(\frac{\ly}{\lx} \frac{\snrbob}{\snreve}\right)
	\end{equation}
	for which $\underline{\varepsilon}=F_{\Xt}(s)$. This means that we can ensure that the lower bound is the best possible by choosing $\ratesecret$ according to \eqref{eq:cond-rs-best-lower-bound}.
\end{rem}

Another interesting aspect is the asymptotic behavior of the outage probability. In \cite{Bloch2008b}, it is shown that for independent channels, the outage probability tends to a positive value for $\ratesecret\to 0$. However, it is possible that the lower bound converges to zero. Namely, if $\ly/\snreve > \lx/\snrbob$, the second term of \eqref{eq:lower-rayleigh-csit-main} holds which goes to zero for $\ratesecret\to 0$ (recall that $s=2^{\ratesecret}-1$). Otherwise, the lower bound also approaches a positive limit which is given in \eqref{eq:lower-rayleigh-main-csit-limit-rs0} {at the bottom of the page}.
\begin{figure*}[b]
	\noindent\rule{\textwidth}{.5pt}
	\begin{equation}\label{eq:lower-rayleigh-main-csit-limit-rs0}
	\lim\limits_{\ratesecret\to 0} \underline{\varepsilon_\text{CSIT}} = \begin{cases}
	\exp\left(\frac{\ly\log\frac{\ly\snrbob}{\lx\snreve}}{\snreve\left(\frac{\lx}{\snrbob}-\frac{\ly}{\snreve}\right)}\right) - \exp\left(\frac{\lx\log\frac{\ly\snrbob}{\lx\snreve}}{\snrbob\left(\frac{\lx}{\snrbob}-\frac{\ly}{\snreve}\right)}\right) & \text{\upshape if }\; \dfrac{\ly}{\snreve} < \dfrac{\lx}{\snrbob}
	\\
	0 & \text{\upshape else}
	\end{cases}
	\end{equation}
\end{figure*}
For $\ratesecret\to\infty$, the outage probabilities both in the independent and best case go to one, i.e., confidential transmission is impossible.

Another asymptotic behavior of interest is the diversity gain $d$ which is defined as~\cite[Def.~1]{LizhongZheng2003}
\begin{equation}\label{eq:def-diversity-gain}
d = \lim\limits_{\snrbob\to\infty} -\frac{\log\varepsilon}{\log\snrbob}\,.
\end{equation}
For fixed $\lyt$ and $\snrbob\to\infty$, the lower bound $\underline{\varepsilon}$ from \eqref{eq:lower-rayleigh-csit-main} is given as $F_{\Xt}(s) = 1-\exp(-\lx s/\snrbob)$ and the diversity gain is therefore
\begin{equation*}
\underline{d}_{\text{CSIT}} = \lim\limits_{\snrbob\to\infty} -\frac{\log\underline{\varepsilon_{\text{CSIT}}}}{\log\snrbob} = \lim\limits_{\snrbob\to\infty} -\frac{\log F_{\Xt}(s)}{\log\snrbob} = 1\,.
\end{equation*}

\subsubsection{Upper Bound}
Since the stationary points are the same for both the lower and upper bound, we can build on the results derived for the lower bound to obtain the upper bound. From \eqref{eq:upper-bound-general-csit}, we know that the upper bound is one, if $h$ has a maximum at $\yt^\star$. From the previous results, we know that this is the case for $\lxt>\lyt$. For $\lxt<\lyt$, it can easily be verified that $\yt^\star<0$.
All of this can be combined to the upper bound on the secrecy outage probability for Rayleigh fading as
\begin{equation}\label{eq:upper-rayleigh-csit-main}
\overline{\varepsilon_\text{CSIT}} = \begin{cases}
1 & \text{for } \lxt \geq \lyt\\
h\left(\frac{\lxt s + \log\frac{\lyt}{\lxt}}{\lxt-\lyt}\right) & \text{for } \lxt < \lyt
\end{cases}
\end{equation}
with
\begin{equation}\label{eq:h-rayleigh}
h(\yt) = 1 - \exp(\lxt(\yt-s)) + \exp(\lyt\yt)\,.
\end{equation}

Note that $\ly/(2^{\ratesecret}\snreve)=\lyt<\lxt=\lx/\snrbob$ is a characterization that Eve's channel is better than Bob's channel. An example when this can occur is if $\lx=\ly$ and Eve has a higher \gls{snr} than Bob, i.e., $\snreve>\snrbob$. In this case, the upper bound on the secrecy outage probability is one.

Obviously, for $\ratesecret\to\infty$ the upper bound on the outage probability also goes to one. For $\ratesecret\to 0$, the same structure as for the lower bound holds, cf.~\eqref{eq:lower-rayleigh-main-csit-limit-rs0}. The only difference is that the function $h$ needs to be used for the evaluation.

Since $\overline{\varepsilon}=h(\yt^\star)$ for $\snrbob\to\infty$, the diversity gain of the upper bound is given as
\begin{equation*}
\overline{d}_{\text{CSIT}} = \lim\limits_{\snrbob\to\infty} -\frac{\log\overline{\varepsilon_{\text{CSIT}}}}{\log\snrbob} = \lim\limits_{\snrbob\to\infty} -\frac{\log h(\yt^\star)}{\log\snrbob} = 1\,.
\end{equation*}

\subsubsection{Independent Channels}
The outage probability for independent channels is evaluated according to \eqref{eq:outage-independent-main-csit}. In the specific case of Rayleigh fading, this is
\begin{equation}\label{eq:outage-rayleigh-independent-main-csit}
\varepsilon_{\text{ind, CSIT}} = 1 - \frac{\lyt\exp\left(-\lxt s\right)}{\lyt+\lxt}\,.
\end{equation}
As expected, this is identical to the result in \cite[Prop.~2]{Bloch2008b}.
With a reminder that $s=2^{\ratesecret}-1$ and $\lyt=\ly/(2^{\ratesecret}\snreve)$, it is easy to see from \eqref{eq:outage-rayleigh-independent-main-csit} that $\varepsilon_{\text{ind}}\to 1$ for $\ratesecret\to\infty$ and $\varepsilon_{\text{ind}}=\frac{\lxt}{\lxt+\ly/\snreve}$ for $\ratesecret\to 0$.

Another interesting observation, which is mentioned in \cite{Bloch2008b}, is that $\varepsilon_{\text{ind}}\approx 1-\exp(-\lxt s)$ for $\lx/\snrbob\ll\ly/\snreve$. In \eqref{eq:lower-rayleigh-csit-main}, we showed that this corresponds to the lower bound on $\varepsilon$, i.e., the independent case approaches the best case.
From this, it is clear that the diversity gain for independent channels is also one.

\subsubsection{Numerical Example}
In the following, we provide some numerical examples for the results derived in the previous section. We will fix the fading coefficients of the Rayleigh fading to be the same as $\lx=\ly=1$. However, all plots can be found as interactive versions online at \cite{BesserGitlab} where these parameters can also be varied. We encourage the reader to try further parameter combinations on their own.

As a first example, Fig.~\ref{fig:rayleigh-csit-main} shows the upper and lower bounds on the secrecy outage probability along with the independent case over different values of Bob's \gls{snr}. All curves are plotted for two different \glspl{snr} of the eavesdropper, \SIlist{0;10}{\decibel}. The secrecy rate is set to $\ratesecret=0.1$. {For verification of the theoretical results, we also show results obtained using Monte Carlo simulations with $10^5$ samples for each point. The source code can be found at \cite{BesserGitlab}.}
As expected, all outage probabilities decrease in general when Eve's \gls{snr} decreases. The only exception happens when the lower bound is above a certain \gls{snr} of the main channel. After the point at which the lower bound \enquote{switches}, outages are only caused by Bob's channel. Therefore, Eve's \gls{snr} has no influence on it. From the condition in \eqref{eq:lower-rayleigh-csit-main}, we know that the needed \gls{snr} $\snrbob$, at which the switch occurs, increases with an increasing \gls{snr} of the eavesdropper $\snreve$. In the plot, this means that the \enquote{switching-point} occurs at around \SI{0}{\decibel} for $\snreve=\SI{0}{\decibel}$ and around \SI{10.5}{\decibel} for $\snreve=\SI{10}{\decibel}$. The upper bound shows a similar behavior. Up to a certain \gls{snr} of the main channel, it is constantly \num{1}. Above this \gls{snr}, it decreases, cf.~\eqref{eq:upper-rayleigh-csit-main}.
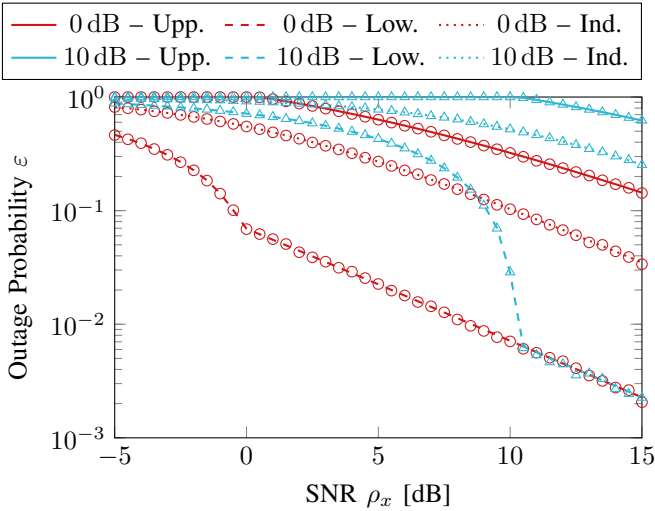
\begin{figure}
	\centering
	\input{img/plot-rayleigh-main-csit.tex}
	
	\vspace*{-1.2em}
	\caption{Secrecy outage probabilities over different \gls{snr} values $\snrbob$ of the main channel for Rayleigh fading channels and different values of the eavesdropper's \gls{snr} $\snreve$. The transmitter has perfect \gls{csi} about the main channel. The parameters are $\lambda_{x}=\lambda_{y}=1$ and $\ratesecret=0.1$. {The markers indicate results obtained by Monte Carlo simulations.}}
	\label{fig:rayleigh-csit-main}
\end{figure}

The next example illustrates the different behaviors with respect to the {$\varepsilon$-outage secrecy rate, i.e., the maximum rate for which the secrecy outage probability is at most $\varepsilon$}.
{In Fig.~\ref{fig:rayleigh-csit-main-eps-outage-rates}, the upper and lower bounds and the independent case are shown over the tolerated outage probability $\varepsilon$. The values are obtained by numerically solving the expressions of the different outage probabilities, e.g., \eqref{eq:upper-rayleigh-csit-main}, for $\ratesecret$. The source code can be found at \cite{BesserGitlab}.}
The \gls{snr} of the main channel is fixed to $\snrbob=\SI{5}{\decibel}$. For the first three curves the \gls{snr} of the eavesdropper is set to $\snreve=\SI{0}{\decibel}$, i.e., it is smaller than Bob's \gls{snr}.
{First, it can be seen that the $\varepsilon$-outage secrecy rates converge to infinity for $\varepsilon\to 1$. This is expected since $\varepsilon=1$ corresponds to the case that one tolerates an arbitrary number of outages. Furthermore, for both the worst case and independent case, there exists a positive value of $\varepsilon$ below which the highest achievable secrecy outage rate is zero, i.e., no secure transmission is possible. These minimum values of $\varepsilon$ are indicated by the dashed lines in Fig.~\ref{fig:rayleigh-csit-main-eps-outage-rates}. On the other hand, it is possible to achieve positive $\varepsilon$-outage secrecy rates for arbitrarily small $\varepsilon$ in the best case. However, this is only possible for the case that Bob's channel is better than Eve's, i.e., $\lxt\geq\lyt$, cf.~\eqref{eq:lower-rayleigh-main-csit-limit-rs0}. Otherwise, the minimum $\varepsilon$, for which the highest achievable secrecy rate $\ratesecret$ is positive, is also positive. This effect can be seen in the example, where we set $\snreve=\SI{5.1}{\decibel}$. As expected, the $\varepsilon$-outage secrecy rates in the independent case decrease when increasing Eve's \gls{snr}. For this example, no secure transmission is possible in the worst case, if $\snreve>\snrbob$, i.e., $\ratesecret=0$.}
{
\begin{figure}
	\centering
	\input{img/plot-rayleigh-main-csit-eps_outage_rates.tex}
	
	\vspace*{-0.5em}
	\caption{$\varepsilon$-outage secrecy rates over the tolerated outage probability $\varepsilon$ for Rayleigh fading channels. The transmitter has perfect \gls{csi} about the main channel. The parameters are $\snrbob=\SI{5}{\decibel}$, and $\lambda_{x}=\lambda_{y}=1$. The first three curves are for Eve's \gls{snr} $\snreve=\SI{0}{\decibel}$. The last two curves are the best and independent case for $\snreve=\SI{5.1}{\decibel}$, respectively. The dashed lines indicate the minimum values of $\varepsilon$ for which positive secrecy rates are achievable.}
	\label{fig:rayleigh-csit-main-eps-outage-rates}
\end{figure}
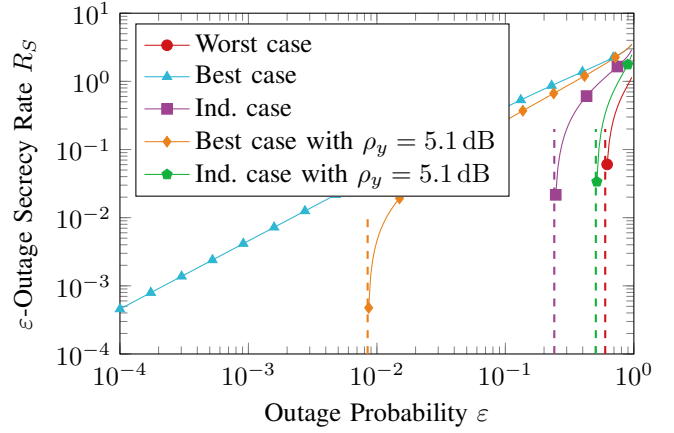
}

In the final example, we want to further discuss the behavior of the secrecy outage probabilities for changing \gls{snr} of the eavesdropper $\snreve$.
In Fig.~\ref{fig:rayleigh-csit-main-vs-snreve}, the bounds and the independent case are shown over different values of $\snreve$. The secrecy rate is set to $\ratesecret=0.1$. The \gls{snr} of the main channel is fixed to $\snrbob=\SI{15}{\decibel}$.
The first noticeable behavior is that all probabilities approach the lower bound for $\snreve\to-\infty$. Especially for the independent case, this is what we expected based on the discussion in the previous subsection. Another interesting behavior is the already mentioned \enquote{switching-point} of the lower bound. Up to a certain \gls{snr}, the lower bound is constant and independent of $\snreve$. Once Eve's \gls{snr} is above a certain value, the lower bound increases with increasing $\snreve$. The specific value of $\snreve$ can be determined according to \eqref{eq:lower-rayleigh-condition-snr-eve-main-csit} and is shown as a dashed line in the plot.
Please note that in Fig.~\ref{fig:rayleigh-csit-main-vs-snreve}, this switching-point is at approximately Bob's \gls{snr}, $\snreve\approx\SI{14.7}{\decibel}$. This is not true in general but only for the specific parameters we chose. For the general condition, please see \eqref{eq:lower-rayleigh-condition-snr-eve-main-csit}. Additionally, all plots are available interactively at \cite{BesserGitlab}. We encourage the readers to change the different parameters on their own and observe their influence on the mentioned points.
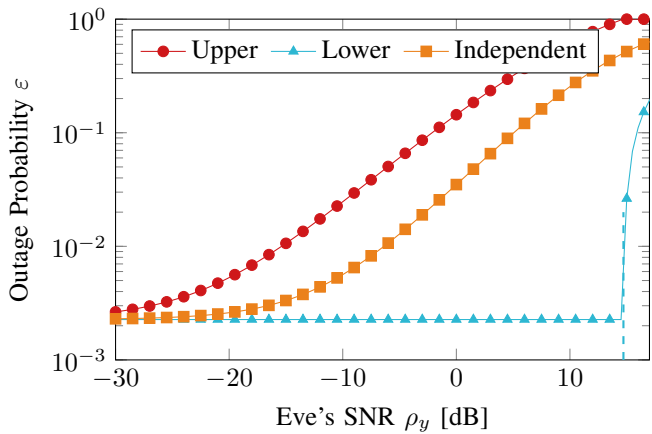
\begin{figure}
	\centering
	\input{img/plot-rayleigh-main-csit-vs-snreve.tex}
	
	\vspace*{-0.5em}
	\caption{Secrecy outage probabilities over the eavesdropper's \gls{snr} $\snreve$ for Rayleigh fading channels. The transmitter has perfect \gls{csi} about the main channel. The parameters are $\snrbob=\SI{15}{\decibel}$, $\lambda_{x}=\lambda_{y}=1$, and $\ratesecret=0.1$. The dashed line indicates Eve's \gls{snr} according to \eqref{eq:lower-rayleigh-condition-snr-eve-main-csit} at which the lower bound changes.}
	\label{fig:rayleigh-csit-main-vs-snreve}
\end{figure}

\subsection{Statistical Channel State Information at the Transmitter}\label{sub:rayleigh-no-csit}
We now drop the assumption of having perfect \gls{csit} and apply the results from Section~\ref{sec:outage-no-csit} for the case of Rayleigh fading.

\subsubsection{Lower Bound}
We start with the general expression of $\underline{\varepsilon}$ from \eqref{eq:lower-bound-no-csit} in Theorem~\ref{thm:lower-bound-general-no-csit} and apply the results from Section~\ref{sub:lower-rayleigh-main-csit}.
We know the expression for $\yt^\star$ from \eqref{eq:yt-star-rayleigh} and that $g_1$ only as a maximum, if $\lxt>\lyt$. The only difference to the case of perfect \gls{csit} is the range of $\yt$. Now, the maximum needs to be in $\yt<s-t$, otherwise the lower bound is only dependent on $F_{\Xt}$. For $\lxt>\lyt$, $\yt^\star<s-t$ holds, if additionally $\lyt < \lxt\exp(t(\lyt-\lxt) - \lyt s)$ holds. Otherwise, we immediately get $\underline{\varepsilon}=F_{\Xt}(t)$.

For Rayleigh fading with only statistical \gls{csit}, the secrecy outage probability is therefore lower bounded by
\begin{equation}\label{eq:lower-rayleigh-no-csit}
\begin{split}
\underline{\varepsilon_\text{no}} = \max\left[
g_1\left(\min\left[\frac{\lxt s + \log\frac{\lyt}{\lxt}}{\lxt-\lyt}, s-t\right]\right),\right.&\\1 - \exp(-\lxt t&)\Bigg]\,,
\end{split}
\end{equation}
where $g_1$ is given by \eqref{eq:g-rayleigh}. Also, we want to recall the shorthands $s=2^{\ratesecret}-1$ and $t=2^{\rateconfusion+\ratesecret}-1$.

Similarly to the case where perfect \gls{csit} about the main channel is available, the lower bound is given by $F_{\Xt}(t)$ for large values of $\snrbob$. We therefore get the same diversity
\begin{equation*}
\underline{d}_{\text{no}} = \lim\limits_{\snrbob\to\infty} -\frac{\log\underline{\varepsilon_{\text{no}}}}{\log\snrbob} = \lim\limits_{\snrbob\to\infty} -\frac{\log F_{\Xt}(t)}{\log\snrbob} = 1\,.
\end{equation*}

For very small $\ratesecret$, we get \eqref{eq:limit-rayleigh-lower-no-csit-rs0} {at the bottom of the page}.
\begin{figure*}[!b]
	\noindent\rule{\textwidth}{.5pt}
	\begin{equation}
	\lim\limits_{\ratesecret\to 0} \underline{\varepsilon_\text{no}} = \max\left[g_1\left(\frac{\log\frac{\ly\snrbob}{\lx\snreve}}{\frac{\lx}{\snrbob}-\frac{\ly}{\snreve}}\right), 1-\exp\left(-\lxt(2^{\rateconfusion}-1)\right)\right]\label{eq:limit-rayleigh-lower-no-csit-rs0}
	\end{equation}
\end{figure*}

\subsubsection{Upper Bound}
The upper bound is derived from \eqref{eq:upper-bound-no-csit} in Theorem~\ref{thm:upper-bound-general-no-csit}. We combine this with the previous results and discussions about the optimal $\yt^\star$ for Rayleigh fading.
Recall that a minimum of
\begin{equation*}
h_1(\yt) = 1 - \exp\left(-\lxt(s-\yt)\right) + \exp(\lyt \yt)
\end{equation*}
is attained for $\lxt<\lyt$. Second, we need to minimize $h_1$ over $\yt<s-t$. With reference to Remark~\ref{rem:upper-bound-no-csit-single-ytstar}, we know that the minimum is attained at $\min\left[\yt^\star, s-t\right]$, which is equal to $s-t$, if the additional condition $\lyt < \lxt\exp(t(\lyt-\lxt) - \lyt s)$ holds.
If $\lxt\geq\lyt$, the upper bound is constantly one.
Combining all of this gives the following expression for the upper bound on the secrecy outage probability when only statistical \gls{csit} is available
\begin{equation}
\overline{\varepsilon_\text{no}} = \begin{cases}
h_1\left(\min\left[\frac{\lxt s + \log\frac{\lyt}{\lxt}}{\lxt-\lyt}, s-t\right]\right) & \text{for } \lyt > \lxt\\
1 & \text{for } \lyt\leq\lxt
\end{cases}\,.
\end{equation}
For increasing \gls{snr} values $\snrbob$ of Bob, $\yt^\star$ is decreasing and becomes smaller than $s-t$. Therefore, the diversity of the upper bound is given as
\begin{equation*}
\overline{d}_{\text{no}} = \lim\limits_{\snrbob\to\infty} -\frac{\log\overline{\varepsilon_{\text{no}}}}{\log\snrbob} = \lim\limits_{\snrbob\to\infty} -\frac{\log h_1(\yt^\star)}{\log\snrbob} = 1\,.
\end{equation*}

\subsubsection{Independent Case}
For comparison, we derive the outage probability of independent $\X$ and $\Y$. In this case, the outage probability $\varepsilon_{\text{ind}}$ can be calculated using \eqref{eq:outage-independent-no-csit}. In the case of Rayleigh fading, this evaluates to
\begin{equation}
\varepsilon_{\text{ind, no}} = 1 - \exp\left(-\lxt t\right) + \frac{\lxt\exp\left(\lyt(s-t) - \lxt t\right)}{\lxt+\lyt}\,.
\end{equation}
It is easy to verify that the diversity is also \num{1}.

The limit for small $\ratesecret$ is given as \eqref{eq:limit-rayleigh-indep-no-csit-rs0} {at the top of the next page.}
\begin{figure*}[!t]
	\begin{equation}\label{eq:limit-rayleigh-indep-no-csit-rs0}
	\lim\limits_{\ratesecret\to 0} {\varepsilon_\text{ind,no}} = 1 - \exp\left(-\lxt(2^{\rateconfusion}-1)\right) + \frac{\lxt\exp\left(-(\lxt+\lyt)(2^{\rateconfusion}-1)\right)}{\lxt+\lyt}
	\end{equation}
	\noindent\rule{\textwidth}{.5pt}
\end{figure*}

\subsubsection{Numerical Example}
The bounds and the independent case are shown in Fig.~\ref{fig:rayleigh-no-csit}. {Together with the theoretical results, we show results obtained by Monte Carlo simulations for verification.} The parameters of the transmission and the Rayleigh fading are $\ly=\lx=1$, $\ratesecret=0.1$, and $\rateconfusion=1$. The x-axis shows Bob's \gls{snr} $\snrbob$. All curves are plotted for two different \glspl{snr} of Eve $\snreve$, namely \SIlist{0;10}{\decibel}.
It can be seen that the shapes of the curves look similar to the ones where perfect \gls{csit} is available, cf. Fig~\ref{fig:rayleigh-csit-main}. As expected, all outage probabilities increase with an increase in Eve's \gls{snr} $\snreve$. Only the lower bound is independent of $\snreve$ above a certain threshold $\snrbob$.
\begin{figure}
	\centering
	\input{img/plot-rayleigh-no-csit.tex}
	\vspace*{-1.5em}
	\caption{Secrecy outage probabilities over different \gls{snr} values $\snrbob$ of the main channel for Rayleigh fading channels and different values of the eavesdropper's \gls{snr} $\snreve$. The transmitter has only statistical \gls{csi}. The parameters are $\lambda_{x}=\lambda_{y}=1$, $\ratesecret=0.1$, and $\rateconfusion=1$. {The markers indicate results obtained by Monte Carlo simulations.}}
	\label{fig:rayleigh-no-csit}
\end{figure}

The last example compares the two scenarios of perfect \gls{csit} and statistical \gls{csit} of the main channel. In Fig.~\ref{fig:rayleigh-comparison-no-csit-main-csit}, the different curves for both scenarios are shown. The \gls{snr} of the eavesdropper is fixed to $\snreve=\SI{5}{\decibel}$. The rates are set to $\ratesecret=0.5$ and $\rateconfusion=1$. Note that only the case of having statistical \gls{csit} depends on $\rateconfusion$. As expected, the outage probabilities are generally lower when perfect \gls{csit} is available. Only the upper bound is identical. However, this is only due to the chosen values of the different parameters. It is possible to obtain greater values for the upper bound with statistical \gls{csit} than with perfect \gls{csit} by increasing $\rateconfusion$. The plot is also available as an interactive version at \cite{BesserGitlab}, where a curious reader can observe this behavior.
\begin{figure}
	\centering
	\input{img/plot-rayleigh-comparison-no-csit-main-csit.tex}
	\vspace*{-1.5em}
	\caption{Comparison of the secrecy outage probabilities over different \gls{snr} values $\snrbob$ of the main channel for Rayleigh fading channels with perfect and with only statistical \gls{csit}. The parameters are $\snreve=\SI{5}{\decibel}$, $\lambda_{x}=\lambda_{y}=1$, $\ratesecret=0.5$, and $\rateconfusion=1$. {The markers indicate results obtained by Monte Carlo simulations.}}
	\label{fig:rayleigh-comparison-no-csit-main-csit}
\end{figure}
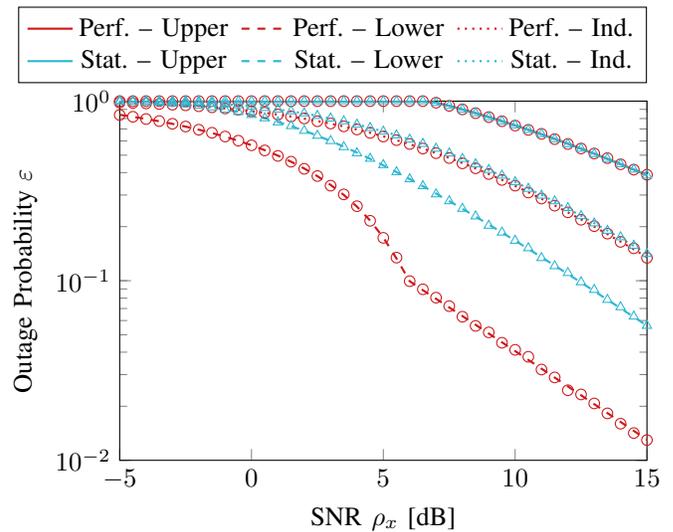

\subsection{Alternative Secrecy Outage Definition}
For comparison, we now take a look at the secrecy outage probability according to the alternative definition from \cite{Zhou2011} discussed in Section~\ref{sec:full-outage}.
The bounds are evaluated for Rayleigh fading according to Theorems~\ref{thm:bounds-general-full-outage-main-csit} and \ref{thm:bounds-general-full-outage-no-csit}.

First, we give an example for the case that only statistical \gls{csit} is available. From Theorem~\ref{thm:bounds-general-full-outage-no-csit}, we know that the outage probability is given as the dual of the copula in this case. Figure~\ref{fig:rayleigh-full-outage-no-csit} shows the upper and lower bound together with the independent case for Rayleigh fading with $\ratesecret=0.1$, $\rateconfusion=1$, and $\snreve=\SI{0}{\decibel}$. The probabilities of the individual events $\eventbob$ and $\eventeve$ are also shown for comparison.
The first observations are about the probabilities of the events $\eventbob$ and $\eventeve$. As expected, the outage probability at Bob decreases with his \gls{snr} $\snrbob$ increasing. The outage event $\eventeve$, due to Eve being able to decode, remains constant since we set her \gls{snr} $\snreve$ to a constant value. Next, recall that the minimum probability of the union of events is the maximum of the probabilities of the individual events.
From Fig.~\ref{fig:rayleigh-full-outage-no-csit}, it can be seen that the lower bound achieves this trivial bound. Thus, it first reduces with increasing $\snrbob$ until the probability from $\eventeve$ becomes larger. It then remains constant when further increasing $\snrbob$. Both the upper bound and the independent case approach the lower bound for high values of Bob's \gls{snr} $\snrbob$.
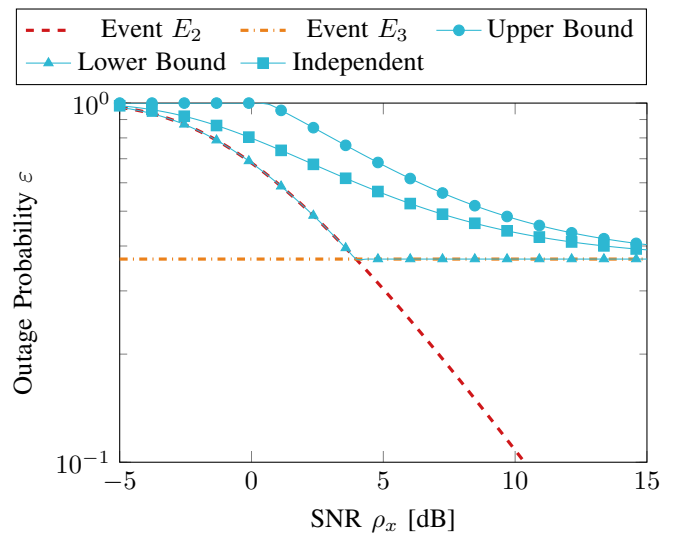
\begin{figure}
	\centering
	\input{img/plot-rayleigh-full-outage-no-csit.tex}
	\vspace*{-1.5em}
	\caption{Secrecy outage probabilities from \eqref{eq:def-event-full-outage-no-csit} over different \gls{snr} values $\snrbob$ of the main channel for Rayleigh fading channels with only statistical \gls{csit}. The parameters are $\snreve=\SI{0}{\decibel}$, $\lambda_{x}=\lambda_{y}=1$, $\ratesecret=0.1$, and $\rateconfusion=1$.}
	\label{fig:rayleigh-full-outage-no-csit}
\end{figure}

Next, we compare the scenarios that Alice has perfect \gls{csi} about the main channel and that only statistical \gls{csit} is available. The bounds and the independent case are shown for Rayleigh fading in Fig.~\ref{fig:rayleigh-full-outage-comparison-no-csit-main-csit}. The first interesting observation is that both outage probabilities are the same in the worst-case, i.e., the upper bound. For both the best-case and the independent case, there is an improvement in the outage probability when perfect \gls{csit} is available. However, all curves approach the same fundamental limit for increasing $\snrbob$, due to Eve being able to decode. Thus, the advantage of having perfect \gls{csit} vanishes in the high-\gls{snr} regime.
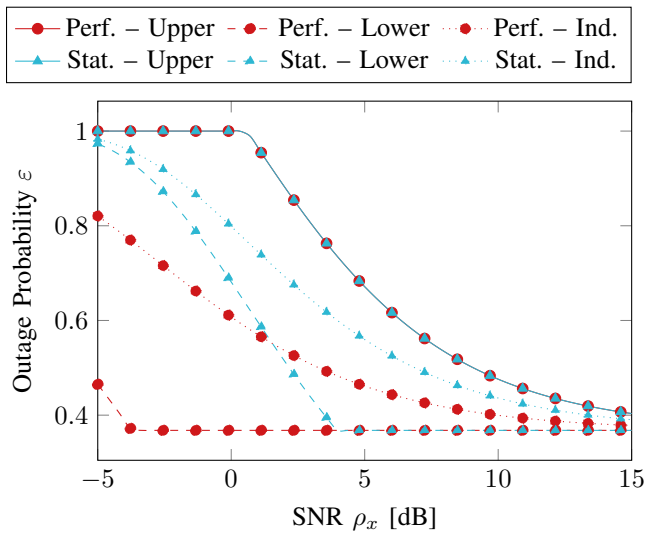
\begin{figure}
	\centering
	\input{img/plot-rayleigh-full-outage-comparison-no-csit-main-csit.tex}
	\vspace*{-0.5em}
	\caption{Comparison of the secrecy outage probabilities from \eqref{eq:def-event-full-outage-no-csit} over different \gls{snr} values $\snrbob$ of the main channel for Rayleigh fading channels with perfect and with only statistical \gls{csit}. The parameters are $\snreve=\SI{0}{\decibel}$, $\lambda_{x}=\lambda_{y}=1$, $\ratesecret=0.1$, and $\rateconfusion=1$.}
	\label{fig:rayleigh-full-outage-comparison-no-csit-main-csit}
\end{figure}

%% file: img/lower-bound-monotonic.tex
\begin{tikzpicture}%
\begin{axis}[
	width=.98\linewidth,
	height=.25\textheight,
	xlabel={$\yt$},
	ylabel={$g(\yt)$},
	domain=-5:0,
	xmin=-5,
	xmax=0,
	legend style = {
		at = {(0.01, 0.52)},
		anchor = west,
	},
	legend cell align=left,
]
\addplot[thick, plot0] {-exp(-(1-x))+exp(2*x)};
\addlegendentry{$\lyt>\lxt$};

\addplot[thick, plot1] {-exp(-(1-x))+exp(.5*x)};
\addlegendentry{$\lxt\exp(-\lxt s) < \lyt < \lxt$};

\addplot[thick, plot3] {-exp(-(1-x))+exp(.1*x)};
\addlegendentry{$\lyt < \lxt\exp(-\lxt s)$};

\addplot[black, thin, dashed] {0};

\addplot[plot0, mark=*] coordinates {(-1.6931471805599454, -0.03383382080915317)};
\addplot[plot1, mark=*] coordinates {(0, 0.6321205588285577)};
\addplot[plot3, mark=*] coordinates {(-1.4473167699933838, 0.778728986667413)};
\end{axis}
\end{tikzpicture}

%% file: img/joint-pdf-rayleigh-lower-main-csit.tex
\begin{tikzpicture}
\begin{axis}[
	width=.98\linewidth,
	height=.26\textheight,
	xlabel={$\X$},
	ylabel={$\Y$},
	view={0}{90},
	colorbar,
	colorbar style={
		at={(.9,.95)},
		anchor=north east,
		height=0.65*\pgfkeysvalueof{/pgfplots/parent axis height},
	},
	colormap name=colormap/cool,
	point meta max=3,
	point meta min=0,
	xmin=0,
	xmax=2,
	ymin=0,
	ymax=2,
	domain=-2:2,
]
    \addplot3[contour filled={number=50},
    		  mesh/rows=50,
    		  mesh/cols=50,
    		  patch type=bilinear,
    		 ] table {data/joint_pdf_main_csit-bob0-eve0-rs1.dat};
\end{axis}
\end{tikzpicture}

%% file: img/plot-rayleigh-main-csit.tex
\begin{tikzpicture}%
\begin{axis}[
	width=.97\linewidth, %
	height=.25\textheight,
	xmin=-5,
	xmax=15,
	ymax=1,
	ymin=.001,
	ymode=log,
	xlabel={SNR $\snrbob$ [dB]},
	ylabel={Outage Probability $\varepsilon$},
	legend cell align=center,
	legend columns=3,
	legend pos=south west,
	legend style = {
		at = {(1, 1.05)}, %
		anchor = south east, %
		/tikz/every even column/.append style={column sep=0.1cm}
	},
]
\addplot[plot0, mark=none, thick] table[x=snr,y=upper] {data/secrecy_outage_main_csit-eve_0.0-rs_0.1-lx_1-ly_1-MC.dat};
\addlegendentry{$\SI{0}{\decibel}$ -- Upp.};

\addplot[plot0, dashed, mark=none, thick] table[x=snr,y=lower] {data/secrecy_outage_main_csit-eve_0.0-rs_0.1-lx_1-ly_1-MC.dat};
\addlegendentry{$\SI{0}{\decibel}$ -- Low.};

\addplot[plot0, dotted, mark=none, thick] table[x=snr,y=indep] {data/secrecy_outage_main_csit-eve_0.0-rs_0.1-lx_1-ly_1-MC.dat};
\addlegendentry{$\SI{0}{\decibel}$ -- Ind.};

\addplot[plot1, mark=none, thick] table[x=snr,y=upper] {data/secrecy_outage_main_csit-eve_10.0-rs_0.1-lx_1-ly_1-MC.dat};
\addlegendentry{$\SI{10}{\decibel}$ -- Upp.};

\addplot[plot1, dashed, mark=none, thick] table[x=snr,y=lower] {data/secrecy_outage_main_csit-eve_10.0-rs_0.1-lx_1-ly_1-MC.dat};
\addlegendentry{$\SI{10}{\decibel}$ -- Low.};

\addplot[plot1, dotted, mark=none, thick] table[x=snr,y=indep] {data/secrecy_outage_main_csit-eve_10.0-rs_0.1-lx_1-ly_1-MC.dat};
\addlegendentry{$\SI{10}{\decibel}$ -- Ind.};

\addplot[plot0, only marks, mark=o] table[x=snr,y=upperMC] {data/secrecy_outage_main_csit-eve_0.0-rs_0.1-lx_1-ly_1-MC.dat};
\addplot[plot0, only marks, mark=o] table[x=snr,y=lowerMC] {data/secrecy_outage_main_csit-eve_0.0-rs_0.1-lx_1-ly_1-MC.dat};
\addplot[plot0, only marks, mark=o] table[x=snr,y=indepMC] {data/secrecy_outage_main_csit-eve_0.0-rs_0.1-lx_1-ly_1-MC.dat};

\addplot[plot1, only marks, mark=triangle] table[x=snr,y=upperMC] {data/secrecy_outage_main_csit-eve_10.0-rs_0.1-lx_1-ly_1-MC.dat};
\addplot[plot1, only marks, mark=triangle] table[x=snr,y=lowerMC] {data/secrecy_outage_main_csit-eve_10.0-rs_0.1-lx_1-ly_1-MC.dat};
\addplot[plot1, only marks, mark=triangle] table[x=snr,y=indepMC] {data/secrecy_outage_main_csit-eve_10.0-rs_0.1-lx_1-ly_1-MC.dat};

\end{axis}
\end{tikzpicture}

%% file: img/plot-rayleigh-main-csit-eps_outage_rates.tex
\begin{tikzpicture}%
\begin{axis}[
	width=.95\linewidth,
	height=.25\textheight,
	ymin=0.0001,
	ymax=10,
	xmax=1,
	xmin=.0001,
	ymode=log,
	xmode=log,
	ylabel={$\varepsilon$-Outage Secrecy Rate $\ratesecret$},
	xlabel={Outage Probability $\varepsilon$},
	legend cell align=left,
	legend columns=1,
	legend pos=north west,
	legend style = {
		/tikz/every even column/.append style={column sep=0.1cm}
	},
]
\addplot[plot0, mark=*, mark repeat=15] table[x=eps,y=upper] {data/eps_outage_sec_rates-main_csit-lx1-ly1-snrx5-snry0.dat};
\addlegendentry{Worst case};

\addplot[plot1, mark=triangle*, mark repeat=15] table[x=eps,y=lower] {data/eps_outage_sec_rates-main_csit-lx1-ly1-snrx5-snry0.dat};
\addlegendentry{Best case};

\addplot[plot2, mark=square*, mark repeat=15] table[x=eps,y=indep] {data/eps_outage_sec_rates-main_csit-lx1-ly1-snrx5-snry0.dat};
\addlegendentry{Ind. case};

\addplot[plot3, mark=diamond*, mark repeat=15] table[x=eps,y=lower] {data/eps_outage_sec_rates-main_csit-lx1-ly1-snrx5-snry5.1.dat};
\addlegendentry{Best case with $\snreve=\SI{5.1}{\decibel}$};

\addplot[plot4, mark=pentagon*, mark repeat=15] table[x=eps,y=indep] {data/eps_outage_sec_rates-main_csit-lx1-ly1-snrx5-snry5.1.dat};
\addlegendentry{Ind. case with $\snreve=\SI{5.1}{\decibel}$};
\addplot[plot0, dashed, mark=none, thick] coordinates {(0.5985108496700132, 0.0001) (0.5985108496700132, .2)};
\addplot[plot2, dashed, mark=none, thick] coordinates {(0.24025307335204213, 0.0001) (0.24025307335204213, .2)};
\addplot[plot3, dashed, mark=none, thick] coordinates {(0.0084705500480447, 0.0001) (0.0084705500480447, .01)};
\addplot[plot4, dashed, mark=none, thick] coordinates {(0.5057562084111449, 0.0001) (0.5057562084111449, .2)};

\end{axis}
\end{tikzpicture}

%% file: img/plot-rayleigh-main-csit-vs-snreve.tex
\begin{tikzpicture}%
\begin{axis}[
	width=.98\linewidth,
	height=.25\textheight,
	xmin=-30,
	xmax=17,
	ymax=1,
	ymin=.001,
	ymode=log,
	xlabel={Eve's SNR $\snreve$ [dB]},
	ylabel={Outage Probability $\varepsilon$},
	legend pos=north west,
	legend cell align=center,
	legend columns = 3,
	legend style = {
		/tikz/every even column/.append style={column sep=0.1cm}
	},
]
\addplot[plot0, mark=*, mark repeat=3] table[x={snr_eve},y=upper] {data/secrecy_outage_main_csit-bob_15.0-rs_0.1-lx_1-ly_1.dat};
\addlegendentry{Upper};

\addplot[plot1, mark=triangle*, mark repeat=3] table[x={snr_eve},y=lower] {data/secrecy_outage_main_csit-bob_15.0-rs_0.1-lx_1-ly_1.dat};
\addlegendentry{Lower};

\addplot[plot3, mark=square*, mark repeat=3] table[x={snr_eve},y=indep] {data/secrecy_outage_main_csit-bob_15.0-rs_0.1-lx_1-ly_1.dat};
\addlegendentry{Independent};

\addplot[plot1, dashed, thick] coordinates {(14.708827082706847, 0.001) (14.708827082706847, .02)};
\end{axis}
\end{tikzpicture}

%% file: img/plot-rayleigh-no-csit.tex
\begin{tikzpicture}%
\begin{axis}[
	width=.97\linewidth,
	height=.26\textheight,
	xmin=-5,
	xmax=15,
	ymax=1,
	ymin=.01,
	ymode=log,
	xlabel={SNR $\snrbob$ [dB]},
	ylabel={Outage Probability $\varepsilon$},
	legend pos=south west,
	legend cell align=center,
	legend columns = 3,
	legend style = {
		at = {(1, 1.05)}, %
		anchor = south east, %
		/tikz/every even column/.append style={column sep=0.1cm}
	},
]
\addplot[plot0, mark=none, thick] table[x=snr,y=upper] {data/secrecy_outage_no_csit-eve_0.0-rs_0.1-lx_1-ly_1-MC.dat};
\addlegendentry{$\SI{0}{\decibel}$ -- Upp.};

\addplot[plot0, dashed, mark=none, thick] table[x=snr,y=lower] {data/secrecy_outage_no_csit-eve_0.0-rs_0.1-lx_1-ly_1-MC.dat};
\addlegendentry{$\SI{0}{\decibel}$ -- Low.};

\addplot[plot0, dotted, mark=none, thick] table[x=snr,y=indep] {data/secrecy_outage_no_csit-eve_0.0-rs_0.1-lx_1-ly_1-MC.dat};
\addlegendentry{$\SI{0}{\decibel}$ -- Ind.};

\addplot[plot1, mark=none, thick] table[x=snr,y=upper] {data/secrecy_outage_no_csit-eve_10.0-rs_0.1-lx_1-ly_1-MC.dat};
\addlegendentry{$\SI{10}{\decibel}$ -- Upp.};

\addplot[plot1, dashed, mark=none, thick] table[x=snr,y=lower] {data/secrecy_outage_no_csit-eve_10.0-rs_0.1-lx_1-ly_1-MC.dat};
\addlegendentry{$\SI{10}{\decibel}$ -- Low.};

\addplot[plot1, dotted, mark=none, thick] table[x=snr,y=indep] {data/secrecy_outage_no_csit-eve_10.0-rs_0.1-lx_1-ly_1-MC.dat};
\addlegendentry{$\SI{10}{\decibel}$ -- Ind.};

\addplot[plot0, only marks, mark=o] table[x=snr,y=upperMC] {data/secrecy_outage_no_csit-eve_0.0-rs_0.1-lx_1-ly_1-MC.dat};
\addplot[plot0, only marks, mark=o] table[x=snr,y=lowerMC] {data/secrecy_outage_no_csit-eve_0.0-rs_0.1-lx_1-ly_1-MC.dat};
\addplot[plot0, only marks, mark=o] table[x=snr,y=indepMC] {data/secrecy_outage_no_csit-eve_0.0-rs_0.1-lx_1-ly_1-MC.dat};

\addplot[plot1, only marks, mark=triangle] table[x=snr,y=upperMC] {data/secrecy_outage_no_csit-eve_10.0-rs_0.1-lx_1-ly_1-MC.dat};
\addplot[plot1, only marks, mark=triangle] table[x=snr,y=lowerMC] {data/secrecy_outage_no_csit-eve_10.0-rs_0.1-lx_1-ly_1-MC.dat};
\addplot[plot1, only marks, mark=triangle] table[x=snr,y=indepMC] {data/secrecy_outage_no_csit-eve_10.0-rs_0.1-lx_1-ly_1-MC.dat};

\end{axis}
\end{tikzpicture}

%% file: img/plot-rayleigh-comparison-no-csit-main-csit.tex
\begin{tikzpicture}%
\begin{axis}[
	width=.97\linewidth,
	height=.26\textheight,
	xmin=-5,
	xmax=15,
	ymax=1,
	ymin=.01,
	ymode=log,
	xlabel={SNR $\snrbob$ [dB]},
	ylabel={Outage Probability $\varepsilon$},
	legend pos=south west,
	legend cell align=center,
	legend columns = 3,
	legend style = {
		at = {(1, 1.05)}, %
		anchor = south east, %
		/tikz/every even column/.append style={column sep=0.1cm}
	},
]

\addplot[plot0, mark=none, thick] table[x=snr,y=upper] {data/secrecy_outage_main_csit-eve_5.0-rs_0.5-lx_1-ly_1-MC.dat};
\addlegendentry{Perf. -- Upper};

\addplot[plot0, dashed, mark=none, thick] table[x=snr,y=lower] {data/secrecy_outage_main_csit-eve_5.0-rs_0.5-lx_1-ly_1-MC.dat};
\addlegendentry{Perf. -- Lower};

\addplot[plot0, dotted, mark=none, thick] table[x=snr,y=indep] {data/secrecy_outage_main_csit-eve_5.0-rs_0.5-lx_1-ly_1-MC.dat};
\addlegendentry{Perf. -- Ind.};

\addplot[plot1, mark=none, thick] table[x=snr,y=upper] {data/secrecy_outage_no_csit-eve_5.0-rs_0.5-lx_1-ly_1-MC.dat};
\addlegendentry{Stat. -- Upper};

\addplot[plot1, dashed, mark=none, thick] table[x=snr,y=lower] {data/secrecy_outage_no_csit-eve_5.0-rs_0.5-lx_1-ly_1-MC.dat};
\addlegendentry{Stat. -- Lower};

\addplot[plot1, dotted, mark=none, thick] table[x=snr,y=indep] {data/secrecy_outage_no_csit-eve_5.0-rs_0.5-lx_1-ly_1-MC.dat};
\addlegendentry{Stat. -- Ind.};

\addplot[plot0, only marks, mark=o] table[x=snr,y=upperMC] {data/secrecy_outage_main_csit-eve_5.0-rs_0.5-lx_1-ly_1-MC.dat};
\addplot[plot0, only marks, mark=o] table[x=snr,y=lowerMC] {data/secrecy_outage_main_csit-eve_5.0-rs_0.5-lx_1-ly_1-MC.dat};
\addplot[plot0, only marks, mark=o] table[x=snr,y=indepMC] {data/secrecy_outage_main_csit-eve_5.0-rs_0.5-lx_1-ly_1-MC.dat};

\addplot[plot1, only marks, mark=triangle] table[x=snr,y=upperMC] {data/secrecy_outage_no_csit-eve_5.0-rs_0.5-lx_1-ly_1-MC.dat};
\addplot[plot1, only marks, mark=triangle] table[x=snr,y=lowerMC] {data/secrecy_outage_no_csit-eve_5.0-rs_0.5-lx_1-ly_1-MC.dat};
\addplot[plot1, only marks, mark=triangle] table[x=snr,y=indepMC] {data/secrecy_outage_no_csit-eve_5.0-rs_0.5-lx_1-ly_1-MC.dat};

\end{axis}
\end{tikzpicture}

%% file: img/plot-rayleigh-full-outage-no-csit.tex
\begin{tikzpicture}%
\begin{axis}[
	width=.97\linewidth,
	height=.26\textheight,
	xmin=-5,
	xmax=15,
	ymax=1,
	ymin=.1,
	ymode=log,
	xlabel={SNR $\snrbob$ [dB]},
	ylabel={Outage Probability $\varepsilon$},
	legend pos=south west,
	legend cell align=center,
	legend columns = 3,
	legend style = {
		at = {(1, 1.05)},
		anchor = south east,
		/tikz/every even column/.append style={column sep=0.1cm}
	},
]

\addplot[plot0, dashed, smooth, very thick] table[x=snr,y=event2] {data/full_secrecy_outage_no_csit-eve_0.0-rs_0.1-rc_1.0-lx_1-ly_1.dat};
\addlegendentry{Event $\eventbob$};

\addplot[plot3, dashdotted, smooth, very thick] table[x=snr,y=event3] {data/full_secrecy_outage_no_csit-eve_0.0-rs_0.1-rc_1.0-lx_1-ly_1.dat};
\addlegendentry{Event $\eventeve$};

\addplot[plot1, mark=*, mark repeat=3, smooth] table[x=snr,y=upper] {data/full_secrecy_outage_no_csit-eve_0.0-rs_0.1-rc_1.0-lx_1-ly_1.dat};
\addlegendentry{Upper Bound};

\addplot[plot1, mark=triangle*, mark repeat=3, smooth] table[x=snr,y=lower] {data/full_secrecy_outage_no_csit-eve_0.0-rs_0.1-rc_1.0-lx_1-ly_1.dat};
\addlegendentry{Lower Bound};

\addplot[plot1, mark=square*, mark repeat=3, smooth] table[x=snr,y=indep] {data/full_secrecy_outage_no_csit-eve_0.0-rs_0.1-rc_1.0-lx_1-ly_1.dat};
\addlegendentry{Independent};

\end{axis}
\end{tikzpicture}

%% file: img/plot-rayleigh-full-outage-comparison-no-csit-main-csit.tex
\begin{tikzpicture}%
\begin{axis}[
	width=.98\linewidth,
	height=.26\textheight,
	xmin=-5,
	xmax=15,
	xlabel={SNR $\snrbob$ [dB]},
	ylabel={Outage Probability $\varepsilon$},
	legend pos=north east,
	legend cell align=center,
	legend columns = 3,
	legend style = {
		at = {(1, 1.05)},
		anchor = south east,
		/tikz/every even column/.append style={column sep=0.1cm}
	},
]

\addplot[plot0, mark=*, mark repeat=3, smooth] table[x=snr,y=upper] {data/full_secrecy_outage_main_csit-eve_0.0-rs_0.1-rc_1.0-lx_1-ly_1.dat};
\addlegendentry{Perf. -- Upper};

\addplot[plot0, dashed, mark=*, mark repeat=3, smooth] table[x=snr,y=lower] {data/full_secrecy_outage_main_csit-eve_0.0-rs_0.1-rc_1.0-lx_1-ly_1.dat};
\addlegendentry{Perf. -- Lower};

\addplot[plot0, dotted, mark=*, mark repeat=3] table[x=snr,y=indep] {data/full_secrecy_outage_main_csit-eve_0.0-rs_0.1-rc_1.0-lx_1-ly_1.dat};
\addlegendentry{Perf. -- Ind.};

\addplot[plot1, mark=triangle*, mark repeat=3, smooth] table[x=snr,y=upper] {data/full_secrecy_outage_no_csit-eve_0.0-rs_0.1-rc_1.0-lx_1-ly_1.dat};
\addlegendentry{Stat. -- Upper};

\addplot[plot1, dashed, mark=triangle*, mark repeat=3, smooth] table[x=snr,y=lower] {data/full_secrecy_outage_no_csit-eve_0.0-rs_0.1-rc_1.0-lx_1-ly_1.dat};
\addlegendentry{Stat. -- Lower};

\addplot[plot1, dotted, mark=triangle*, mark repeat=3, smooth] table[x=snr,y=indep] {data/full_secrecy_outage_no_csit-eve_0.0-rs_0.1-rc_1.0-lx_1-ly_1.dat};
\addlegendentry{Stat. -- Ind.};

\end{axis}
\end{tikzpicture}

%% file: conclusion.tex
\section{Conclusion}\label{sec:conclusion}

In this work, we presented bounds on the secrecy outage probability for wiretap channels with dependent fading coefficients. We always assume statistical \gls{csit} of the eavesdropper channel. The two cases of perfect \gls{csit} about the main channel and only statistical \gls{csit} are considered.
The bounds are compared to the usually considered scenario of independent channels and significant performance differences are shown.
Especially, it is shown that in the best-case, the secrecy outage probability can be independent of Eve's channel. A sufficient condition for this case was derived.

The first motivational question from the introduction can be answered based on this work as follows: \emph{A system designer would first estimate the marginal distributions of the channels, e.g., by taking measurements at different locations. Depending on the available \gls{csit}, he would then apply Theorem~\ref{thm:lower-bound-general-csit} or \ref{thm:upper-bound-general-no-csit} to derive the outage probability in the worst-case. {Solving for $\ratesecret$ then gives the maximum transmission rate at which the outage probability is below the tolerated one, cf. Fig.~\ref{fig:rayleigh-csit-main-eps-outage-rates}.}} %

If the system designer is able to control the joint distribution, he can apply Theorem~\ref{thm:upper-bound-general-csit} or \ref{thm:upper-bound-general-no-csit} to calculate the lower bound on the secrecy outage probability, i.e., the best-case, and use it as a benchmark for his real implementation.

Additionally, it is shown that the copula approach also works for a pessimistic alternative secrecy outage definition. Lower and upper bounds were derived for this case too.

The presented work only deals with a single passive eavesdropper. In future work, this could be extended to consider multiple eavesdroppers. {In both cases of colluding and non-colluding eavesdroppers, a meta-distribution based on the $n$ individual eavesdroppers can be derived~\cite{Tang2018}. However, an extension of the copula bounds to more than two channels is not straightforward. This is mostly due to the fact that the Fr\'{e}chet-Hoeffding lower bound $W$ is not a copula in the case of $n>2$.}
{Another possible extension could be the generalization to imperfect \gls{csit}.}

%% file: app-proofs.tex
\section{Proof of Theorem~\ref{thm:lower-bound-general-csit}\label{app:proof-thm-lower-bound-general-csit}}
First, let us rewrite \eqref{eq:lower-outage-sum} as
\begin{equation}
\underline{\varepsilon} = \sup_{\yt\leq 0}\positive{g(\yt)}
\end{equation}
with the optimization function
\begin{equation}\label{eq:def-g}
g(\yt) = F_{\Xt}(s-\yt) + F_{\Yt}(\yt) - 1\,.
\end{equation}
As a baseline, we take a look at the boundaries,
\begin{align}
\lim\limits_{\yt\to-\infty} g(\yt) &= 0\\
g(0) &= F_{\Xt}(s)\,,
\end{align}
which already shows that the lower bound on the secrecy outage probability can not be lower than $F_{\Xt}(s)$.

We now want to maximize $g$. First, the necessary condition is
\begin{equation}
\frac{\partial g}{\partial \yt} = g^\prime(\yt) = f_{\Yt}(\yt) - f_{\Xt}(s-\yt) \stackrel{!}{=} 0\,.
\end{equation}
The optimal $\yt$ is therefore given by
\begin{equation}\label{eq:condition-yt-star}
f_{\Yt}(\yt^{\star}) = f_{\Xt}(s-\yt^{\star})\,.
\end{equation}
We have a maximum at this point if
\begin{equation}\label{eq:deriv-g-2}
\left.\frac{\partial^2 g}{\partial \yt^2}\right|_{\yt=\yt^\star} = g''(\yt^\star) = f_{\Yt}^{\prime}(\yt^\star) + f_{\Xt}^{\prime}(s-\yt^\star) \stackrel{!}{<} 0\,.
\end{equation}
The lower bound on the outage probability is then given as the maximum over all stationary points $\yt^\star$
\begin{equation*}
\underline{\varepsilon} = \max_{\yt^{\star} \in \mathcal{Y}} g(\yt^{\star})\, ,
\end{equation*}
with constraint set $\mathcal{Y}$ defined in (\ref{eq:cons}). 
We can now combine this to \eqref{eq:lower-bound-general-csit} where we take into account that there might be multiple maximum points.

\section{Proof of Theorem~\ref{thm:upper-bound-general-csit}\label{app:proof-thm-upper-bound-general-csit}}
We start with \eqref{eq:upper-outage-sum} to obtain
\begin{equation}
\overline{\varepsilon} = \inf_{\yt\leq 0} \leqone{h(\yt)}
\end{equation}
with the function
\begin{equation}
h(\yt) = F_{\Xt}(s-\yt) + F_{\Yt}(\yt)\,.
\end{equation}
Again, we take a look at the boundaries first, which are given as
\begin{align}
\lim\limits_{\yt\to-\infty} h(\yt) &= 1\\
h(0) &= 1 + F_{\Xt}(s) \geq 1\,.
\end{align}
This already leads to the conclusion that the upper bound on the outage probability is only less than one, if $h(\yt)$ has a minimum for $\yt<0$.

It is easy to see that the derivations of $h$ are the same as the ones of $g$ from the lower bound. Therefore, the same optimal points $\yt^\star$ specified in \eqref{eq:condition-yt-star} are relevant for our consideration. Combining \eqref{eq:deriv-g-2} with the above observations gives \eqref{eq:upper-bound-general-csit} where the minimum again is over all stationary points $\yt^\star$.

\section{Proof of Theorem~\ref{thm:lower-bound-general-no-csit}\label{app:proof-thm-lower-bound-general-no-csit}}
The general lower bound from \eqref{eq:lower-bound-general-no-csit} can be reformulated as
\begin{equation}
\underline{\varepsilon} = \sup_{\yt\leq 0} g(\yt)
\end{equation}
with the optimization function
\begin{equation}\label{eq:def-g-no-csit}
g(\yt) = \begin{cases}
g_1(\yt) = F_{\Xt}(s-\yt) + F_{\Yt}(\yt) - 1 & \text{for } \yt < s-t\\
g_2(\yt) = F_{\Xt}(t) + F_{\Yt}(\yt) - 1 & \text{for } \yt\geq s-t\\
\end{cases}
\end{equation}
with the shorthand $s=2^{\ratesecret}-1$ and $t=2^{\rateconfusion+\ratesecret}-1$.
The overall solution is then obtained by individually maximizing $g_1$ and $g_2$ and finally taking the maximum between them
\begin{equation}
\underline{\varepsilon} = \max\left[\max_{\yt<s-t} g_1(\yt), \max_{s-t\leq\yt\leq 0} g_2(\yt)\right]\,.
\end{equation}

First, it is easy to see that $g_2$ is increasing in $\yt$ and therefore, the maximum is attained at the maximum $\yt$, i.e. $\yt=0$,
\begin{equation}\label{eq:max-g2-no-csit}
\max_{s-t\leq\yt\leq 0} g_2(\yt) = g_2(0) = F_{\Xt}(t)\,.
\end{equation}

Next, we want to maximize $g_1$ and can observe that it looks identical to $g$ in the case of perfect \gls{csit} from \eqref{eq:def-g}. The only difference is the range of $\yt$, which is $\yt<s-t$ in the case of statistical \gls{csit}.
The boundaries are given as
\begin{align*}
\lim\limits_{\yt\to-\infty} g_1(\yt) &= 0\\
g_1(s-t) &= g_2(s-t)\,,
\end{align*}
where we can see that, 1) the function $g$ is continuous and 2) that the maximum of $g$ can only be in $\yt<s-t$ if $g_1$ has a maximum in this range. Otherwise, the maximum is always given by \eqref{eq:max-g2-no-csit}.
Since the derivatives of $g_1$ are the same as for $g$ from \eqref{eq:def-g}, the conditions on the stationary points are the same, namely \eqref{eq:condition-yt-star} and \eqref{eq:deriv-g-2}.
Combining this gives the lower bound on the secrecy outage probability when only statistical \gls{csit} is available gives \eqref{eq:lower-bound-no-csit}.

\section{Proof of Theorem~\ref{thm:upper-bound-general-no-csit}\label{app:proof-thm-upper-bound-general-no-csit}}
First, we have
\begin{equation}
\overline{\varepsilon} = \inf_{\yt\leq 0} h(\yt) = \min\left[\min_{\yt<s-t} h_1(\yt), \min_{s-t\leq\yt\leq 0} h_2(\yt), 1\right]\,,
\end{equation}
where $h_1$ and $h_2$ are defined in a similar fashion as $g_1$ and $g_2$ from \eqref{eq:def-g-no-csit} but based on the optimization function in \eqref{eq:upper-bound-general-no-csit}.
Equivalent to $g_2$, $h_2$ is increasing in $\yt$ and therefore the minimum is attained at the minimal $\yt$,
\begin{equation}
\min_{s-t\leq\yt\leq 0} h_2(\yt) = h_2\left(s-t\right) = F_{\Xt}\left(t\right) + F_{\Yt}\left(s-t\right).
\end{equation}
Again, $h_1$ has to have a minimum in $\yt<s-t$ in order to attain lower values than $h_2$.
Analogue to \eqref{eq:lower-bound-no-csit}, we combine all of this for the upper bound on $\varepsilon$ to \eqref{eq:upper-bound-no-csit}.

\section{Proof of Theorem~\ref{thm:bounds-general-full-outage-no-csit}\label{app:proof-thm-bounds-general-full-outage-no-csit}}
The probability of $E_{\text{alt,no}}$ from \eqref{eq:def-event-full-outage-no-csit} can be written as
\begin{align}
\varepsilon_{\text{alt,no}} &= \Pr\left(E_{\text{alt,no}}\right)\\
&= \Pr\left(\Xt<t \text{ or } \Yt<s-t\right)\\
&= F_{\Xt}(t) + F_{\Yt}(s-t) - F_{\Xt,\Yt}(t, s-t)\\
&= \bar{C}(F_{\Xt}(t), F_{\Yt}(s-t))\,,
\end{align}
where $\bar{C}$ is the dual of the copula $C$ which determines the joint distribution $F_{\Xt, \Yt}$.
The upper and lower bound on $\varepsilon_{\text{alt,no}}$ follow then immediately from the Fr\'{e}chet-Hoeffding bounds~\cite[Eq.(2.2.5)]{Nelsen2006}
\begin{equation*}
W(a, b) \leq C(a, b) \leq M(a, b)\,.
\end{equation*}